\documentclass[12pt]{amsart}
\usepackage{a4wide}
\usepackage{amssymb}
\usepackage{pdfsync}
\usepackage{wrapfig,epsfig}
\usepackage{color}

\newtheorem{definition}{Definition}[section]
\newtheorem{corollary}[definition]{Corollary}

\newtheorem{lemma}[definition]{Lemma}
\newtheorem{proposition}[definition]{Proposition}
\newtheorem{theorem}[definition]{Theorem}

\makeatletter
\@addtoreset{equation}{section}

\newcommand{\REQ}{\eqref}
\newcommand{\bes}{\begin{displaymath}}
\newcommand{\ees}{\end{displaymath}}
\newcommand{\be}{\begin{equation}}
\newcommand{\ee}{\end{equation}}
\newcommand{\ba}{\begin{eqnarray}}
\newcommand{\ea}{\end{eqnarray}}
\newcommand{\bas}{\begin{eqnarray*}}
\newcommand{\eas}{\end{eqnarray*}}
\newcommand{\@Bbb}[1]{\ensuremath{\Bbb #1}}

\newcommand{\B}{{\@Bbb B}}
\newcommand{\C}{{\@Bbb C}}
\newcommand{\E}{{\@Bbb E}}
\newcommand{\F}{{\@Bbb F}}
\newcommand{\G}{{\@Bbb G}}
\renewcommand{\P}{{\@Bbb P}}
\newcommand{\bbP}{{\P}}
\newcommand{\Q}{{\@Bbb Q}}
\newcommand{\bQ}{{\@Bbb Q}}
\newcommand{\N}{{\@Bbb N}}
\newcommand{\R}{{\@Bbb R}}
\newcommand{\T}{{\@Bbb T}}
\newcommand{\bbR}{{\@Bbb R}}
\newcommand{\W}{{\@Bbb W}}
\newcommand{\Z}{{\@Bbb Z}}
\newcommand{\bbZ}{{\@Bbb Z}}

\newcommand{\si}{\sigma}

\newcommand{\Om}{\Omega}
\newcommand{\om}{\omega}

\newcommand{\@s}[1]{\ensuremath{\mathcal #1}}
\newcommand{\cA}{\@s A}
\newcommand{\cB}{\@s B}
\newcommand{\cC}{\@s C}
\newcommand{\cD}{\@s D}
\newcommand{\cE}{\@s E}
\newcommand{\cF}{\@s F}
\newcommand{\cG}{\@s G}
\newcommand{\cH}{\@s H}
\newcommand{\cI}{\@s I}
\newcommand{\cJ}{\@s J}
\newcommand{\cal}{\mathcal}
\newcommand{\cK}{\@s K}
\newcommand{\cL}{\@s L}
\newcommand{\cN}{\@s N}
\newcommand{\cM}{\@s M}
\newcommand{\cO}{\@s O}
\newcommand{\cP}{\@s P}
\newcommand{\cQ}{\@s Q}
\newcommand{\cR}{\@s R}
\newcommand{\cS}{\@s S}
\newcommand{\cT}{\@s T}
\newcommand{\cU}{\@s U}
\newcommand{\cV}{\@s V}
\newcommand{\cW}{\@s W}
\newcommand{\cX}{\@s X}
\newcommand{\cY}{\@s Y}
\newcommand{\cZ}{\@s Z}

\def\REQ#1{{\rm (\ref{#1})}}
\def\d{{\rm d}}
\def\e{{\rm e}}
\def\i{{\rm i}}

\def\qed{\mbox{$\square$}}
\newcommand{\@bm}[1]{\ensuremath{\mathbf #1}}
\newcommand{\bma}{\@bm a}\newcommand{\bmA}{\@bm A}
\newcommand{\bmb}{\@bm b}\newcommand{\bmB}{\@bm B}
\newcommand{\bmc}{\@bm c}\newcommand{\bmC}{\@bm C}
\newcommand{\bmd}{\@bm d}\newcommand{\bmD}{\@bm D}
\newcommand{\bme}{\@bm e}
\newcommand{\bmf}{\@bm f}\newcommand{\bmF}{\@bm F}
\newcommand{\bmg}{\@bm g}\newcommand{\bmG}{\@bm G}
\newcommand{\bmh}{\@bm h}\newcommand{\bmH}{\@bm H}
\newcommand{\bmi}{\@bm i}\newcommand{\bmI}{\@bm I}
\newcommand{\bmj}{\@bm j}
\newcommand{\bmk}{\@bm k}\newcommand{\bmK}{\@bm K}
\newcommand{\bml}{\@bm l}
\newcommand{\bmm}{\@bm m}\newcommand{\bmM}{\@bm M}
\newcommand{\bmn}{\@bm n}
\newcommand{\bmo}{\@bm o}
\newcommand{\bmp}{\@bm p}
\newcommand{\bmq}{\@bm q}\newcommand{\bmQ}{\@bm Q}
\newcommand{\bmr}{\@bm r}
\newcommand{\bms}{\@bm s}\newcommand{\bmS}{\@bm S}
\newcommand{\bmt}{\@bm t}
\newcommand{\bmu}{\@bm u}\newcommand{\bmU}{\@bm U}
\newcommand{\bmw}{\@bm w}\newcommand{\bmW}{\@bm W}
\newcommand{\bmv}{\@bm v}\newcommand{\bmV}{\@bm V}
\newcommand{\bmx}{\@bm x}\newcommand{\bmX}{\@bm X}\newcommand{\bx}{\@bm x}
\newcommand{\bmy}{\@bm y}\newcommand{\bmY}{\@bm Y}\newcommand{\by}{\@bm y}
\newcommand{\bmz}{\@bm z}\newcommand{\bmZ}{\@bm Z}

\newcommand{\bbE}{\mathbb E}

\newcommand{\bbT}{{\mathbb T}^2}

\newcommand{\bmzero}{\@bm 0}

\newcommand{\ga}{\gamma}

\newcommand{\@g}[1]{\ensuremath{\mathfrak #1}}
\newcommand{\gA}{\@g A}
\newcommand{\gD}{\@g D}
\newcommand{\gJ}{\@g J}
\newcommand{\gF}{\@g F}
\newcommand{\gM}{\@g M}
\newcommand{\gR}{\@g R}

\newcommand{\commentout}[1]{{}}

\makeatother
\begin{document}
\title[Limit theorems for a passive tracer]{Passive tracer in a   flow corresponding to a two dimensional stochastic Navier Stokes equations}
\author{Tomasz Komorowski}
\address{Institute of Mathematics, UMCS, Pl. M. Curie-Sk\l
o\-do\-wskiej 1, 20-031 Lublin, Poland\\
and\\
Institute of Mathematics, Polish Academy of Sciences,
\'Sniadeckich, 8, 00-956, Warsaw, Poland}
\email{komorow@hektor.umcs.edu.pl}
\author{Szymon  Peszat}
\address{Institute of Mathematics, Polish Academy of Sciences,
\'Sw. Tomasza 30/7, 31-027 Krak\'ow, Poland and Faculty of Applied Mathematics, AGH University of Science and Technology, Krak\'ow, Poland}
\email{napeszat@cyf-kr.edu.pl}
\author{Tomasz Szarek}
\address{Institute of Mathematics, University of Gda\'nsk, Wita Stwosza
  57, 80-952 Gda\'nsk, Poland}
\email{szarek@intertele.pl}


\thanks{This work has been partly supported by Polish Ministry of Science and Higher Education  Grants N N201 419139  (T.K. and T.S.),   N N201 419039 (S.P.).  } 
\begin{abstract}

In this paper we prove the law of large numbers and central limit theorem for trajectories of a particle carried by a two dimensional Eulerian velocity field. The field is given by a solution of a stochastic Navier--Stokes system with a non-degenerate noise. The spectral gap property, with respect to Wasserstein metric, for such a system has been shown in \cite{HM1}. In the present paper we show that a similar property holds for the environment process corresponding to the Lagrangian observations of the velocity. In consequence we conclude the law of large numbers and the central limit theorem for the tracer. The proof of the central limit theorem relies on the martingale approximation of the trajectory process.  

\end{abstract}

\date{\today}

\maketitle

\section{Introduction} 

\label{intro}

Consider the Navier--Stokes equations (N.S.E.) on a two dimensional torus $\bbT$,
\begin{equation}
\label{E11a}
\begin{aligned}
&\partial_t\vec u(t,x)+ \vec u(t,x)\cdot \nabla_x\vec u(t,x)= \Delta_x\vec u(t,x )-\nabla_x p(t,x)+ \vec F(t,x),\\
& \nabla \cdot \vec u(t,x)= 0,\\
& \vec u(0,x)=\vec u_0(x).
\end{aligned}
\end{equation}
The two dimensional vector field $\vec u(t,x)$ and scalar  field $p(t,x)$ over $[0,+\infty)\times\bbT$, are  called an Eulerian velocity and pressure, respectively. The forcing $\vec F(t,x)$ is assumed to be a Gaussian white noise in $t$, homogeneous and sufficiently regular in $x$ defined over a certain probability space $(\Om,{\cal F},\bbP)$.
 Consider  the trajectory of a tracer particle defined as the solution of the ordinary differential equation (o.d.e.)
 \begin{equation}
\label{E11b}
\dfrac{d x(t)}{d t}=\vec u(t,x(t)),\quad x(0)=x_0,
\end{equation}
where $x_0\in \bbR^2$. Thanks to well known regularity properties of solutions of N.S.E, see e.g. \cite{MP}, $\vec u(t,x)$ possesses continuous   modification in $x$ for any $t>0$. However, since $\vec u(t,x)$ needs not be Lipschitz in $x$, the  equation might not define  $x(t)$, $t\ge 0$,  as  a stochastic process  over $(\Om,{\cal F},\bbP)$, due to possible non-uniqueness of solutions. In our first result we construct a solution process (see Proposition \ref{CTM})  and show (see Corollary  \ref{unique-law-traj})  that the law of any process  satisfying \eqref{E11b} and adapted to the  natural filtration of   $\vec u$ is uniquely determined.

The main objective of this paper  is to study ergodic properties of the trajectory process. We prove, see part 1) of Theorem \ref{lab3}, the existence of the Stokes drift 
 \begin{equation}\label{eq1}
v_*:=\lim_{t\to+\infty}\frac{x(t)}{t},
\end{equation}
where the limit above is understood in probability.  A similar result for a Markovian and Gaussian velocity field $\vec u$ (that need not be a solution of a N.S.E.) that decorrelates sufficiently fast in time has been considered in  \cite{Kps}. Next, we investigate  the size of ''typical fluctuations'' of the trajectory around its mean. We prove, see part  3) of the theorem, that 
 \begin{equation}
 \label{012603}
 Z(t):=\frac{x(t)-v_*t}{\sqrt{t}}\Rightarrow Z,\quad\mbox{ as } t\to+\infty 
 \end{equation}
where $Z$  is  a random vector with normal distribution ${\cal N}(0,D)$ and the convergence  is understood in law. Moreover, we show that the asymptotic variance of $Z(t)$, as $t\to +\infty$,  exists and coincides with the covariance matrix $D$.  

In our approach a crucial role is played by the {\em Lagrangian process} 
$$
\vec\eta(t,x):=\vec u(t,x(t)+x),\quad  t\ge 0, \ x\in \bbT
$$ 
that describes the environment  from the vantage point of the moving particle.
It turns out that its rotation in $x$, 
$$
\om(t,x)={\rm rot}\,\vec\eta(t,x):=\partial_2\eta_1(t,x)-\partial _1\eta_2(t,x),\quad  t\ge 0,\  x\in \bbT,  
$$
satisfies a stochastic partial differential equation (s.p.d.e.) \eqref{E25a}  that is similar to the stochastic N.S.E. in the vorticity formulation, see \eqref{E25a0}. The position $x(t)$ of the particle  at time $t$, can be represented as an additive functional of the Lagrangian process, i.e. 
$$
x(t)=\int_0^t\psi_*(\om(s))ds,
$$
see the begining of Section \ref{S6} for the definition of $\psi_*$.  Then, \eqref{eq1} and \eqref{012603} become the statements about the law of large numbers and central limit theorem for an additive functional of the process $\eta(\cdot)$.

Following the ideas of Hairer and Mattingly, see \cite{HM,HM1}, we are able to prove, see Theorem  \ref{T1} below, that the transition semigroup of $\om(\cdot)$ satisfies the spectral gap property in a Wasserstein metric defined over the Hilbert space $H$ of  square integrable  mean zero functions. If $\psi_*(\cdot)$ were Lipschitz this fact would make the proof of  the law of large numbers and central limit theorem standard, in view of     \cite{shirikyan} (see also \cite{kowalczuk,kuksin-shirkyan}). However, in our case the observable $\psi_*$  is not Lipschitz. In fact, it is not even defined on the state space $H$ of the process.  Nevertheless, it is  a bounded linear functional over another Hilbert space $V$ that is compactly embedded in $H$. Adopting the approach of Mattingly and Pardoux from \cite{MP},   see Theorem \ref{T2} below, we are able to prove that  the equation for $\om$ has regularization properties similar to  the N.S.E. and that $\om(t)$ belongs to $V$ for any $t>0$. In consequence, one can show that the transition semigroup can be defined on $\psi_*$ and has the same contractive properties as the semigroup defined on Lipschitz functions on $H$. The  law of large numbers can be then shown, Section  \ref{sec5.4},  by a modification of  the argument of Shirikyan    from \cite{shirikyan} (see also \cite{kowalczuk}). To prove the central limit theorem we  construct a corrector field $\chi$, see Section \ref{lab21}, over the ''larger'' space $H$. Then, we  proceed with the classical martingale proof of the central limit theorem, see Section  \ref{sec5.4}. Such an argument has been used to show this type of a theorem for a Lipschitz observable of the solution of a N.S.E. in \cite{shirikyan}. The proof of the existence of the asymptotic variance is done in Section \ref{sec5.3}.

The model  of transport in a fluid flow based on  \eqref{E11b} is  referred to in the literature  as the {\em passive tracer model} (see e.g. Chapter V of \cite{yaglom-monin}). The  $d$-dimensional vector field $\vec u$ appearing on the right hand side of  \eqref{E11b}  is usually assumed  to be    random, stationary,   that in principle may have nothing to do with the N.S.E.  Since the fluid flow is incompressible, equation \eqref{E11b} is complemented by the condition $\nabla_x\cdot \vec u(t,x)\equiv 0$. This model has been introduced by G. Taylor in the 1920-s (see \cite{taylor} and also \cite{kraichnan}) and plays an important role in describing transport phenomena in fluids, e.g.   in investigation of ocean currents (see \cite{stewart}). There exists an extensive literature concerning the  passive tracer both from the mathematical and physical points of view, see e.g. \cite{krama} and the references therein. In particular, it can be shown (see \cite{port-stone}) that the incompressibility assumption implies that the Lagrangian  process $\vec u(t,x(t))$, $t\ge 0$,  is stationary and if one can prove its ergodicity, the Stokes drift coincides with the mean of the field  $\bbE\vec u(0,0)$.  The   weak convergence of $(x(t)-v_*t)/\sqrt{t}$ towards a normal law has  been shown for flows  possessing good relaxation properties either in time, or both in time and space, see \cite{caxu, fk1,kola, koralov} for the  Markovian case, or  \cite{fg-1} for the case of  non-Markovian, Gaussian fields  with finite decorrelation time. According to our knowledge this 
is the first  result when the central limit theorem has been shown for the  tracer  in a flow that is given by an actual solution of the two dimensional N.S.E.

\section{Preliminaries}

\subsection{Some function spaces and operators}

Denote by $\bbT$ the two dimensional torus understood as the product of two segments $[-1/2,1/2]$ with identified endpoints. Trigonometric monomials $e_k(x)=\e^{2i\pi k\cdot x}$, $k=(k_1,k_2)\in\bbZ^2$, form the orthonormal base in  the space $L^2(\bbT)$ of all square integrable functions with the standard scalar product $\langle\cdot,\cdot\rangle$ and norm $|\cdot|$. For a given $w\in L^2(\bbT)$ let $\hat w_k=\langle w,e_k\rangle$. Let $H$ be the subspace of $L^2(\bbT)$  consisting of those functions  $w$, for which $\hat w_0=0$. For any $r\in\bbR$ let 
$$
(-\Delta)^{r/2}w:=\sum_{k\in\bbZ^2_*}|k|^r\hat w_ke_k,\quad w\in H^r,
$$
where $H^r$ consists of such $w$, for which $\sum_{k\in\bbZ^2_*}|k|^{2r}|\hat w_k|^2<+\infty$ and $\bbZ^2_*:=\bbZ^2\setminus\{(0, 0)\}$.  We equip $H^r$  with the graph Hilbert norm  $|\cdot |_r:=|(-\Delta )^{r/2}\cdot |$.   Let $V:=H^1$ and let  $V'$ be the dual to $V$.  Then $H$ can be identified with a subspace of $V'$ and $V\hookrightarrow H \hookrightarrow V'$.  We shall also denote by $\|\cdot\|$ the respective norm $|\cdot|_1$. It is well known (see e.g. Corollary 7.11 of \cite{gilbarg-trudinger}) that $H^{1+s}$ is continuously embedded in $C(\bbT)$ for any $s>0$. Moreover, there exists a constant $C>0$ such that
\begin{equation}
\label{embed}
\|w\|_{\infty}\le C|w|_{1+s},\quad\forall\, w\in C^\infty(\bbT).
\end{equation}
Here 
$\|w\|_{\infty}:=\sup_{x\in\bbT}|w(x)|$. In addition, the following estimate, sometimes referred to as the Gagliardo--Nirenberg inequality, holds, see e.g. p. 27 of \cite{henry}. For any $s>0$, $\beta\in[0,1]$ there exists $C>0$ such that
\begin{equation}
\label{gagliardo-nirenberg}
|w|_{\beta s}\le C|w|^{1-\beta}|w|_{s}^{\beta},\quad\forall\, w\in C^\infty(\bbT).
\end{equation}

Define ${\cal K}\colon H^r\to H^{r+1}\times H^{r+1}$ by
\begin{equation}
\label{010512}
{\cal K}(w)=({\cal K}_1(w),{\cal K}_2(w)):=\sum_{k\in\bbZ^2_*}|k|^{-2}k^\perp \hat w_ke_k.
\end{equation}
We have
\begin{equation}
\label{010512a}
|{\cal K}_i(w)|_{r+1}\le |w|_r, \quad w\in H_r.
\end{equation}
For a given $x\in\bbR^2$ and $w\in H^r$ we let $\tau_xw\in H^r$  be defined by
$$
\tau_xw:=w(\cdot+x)=\sum_{k\in\bbZ^2_*}\e^{-2\pi i k\cdot x}\hat w_ke_k. 
$$

\subsection{Homogeneous Wiener process}

Write 
$$
\mathbb Z^2_+:=[(k_1,k_2)\in \mathbb Z^2_*\colon k_2>0]\cup [(k_1,k_2)\in \mathbb Z^2_*\colon k_1>0,k_2=0]
$$ 
and let $\mathbb Z^2_-:=-\mathbb Z^2_+$.
Let $ (B_k(t))_{t\ge 0}$, $k\in \mathbb Z^2_+$,  be independent, standard  one dimensional Brownian motions defined on a filtered probability space  $(\Omega,\mathcal{F},(\mathcal{F}_t)_{t\ge 0}, \mathbb{P})$. Define $B_{-k}(t):=B_k(t)$ for $k\in \mathbb Z^2_+$.
Assume that the function $k\mapsto q_k$   is  even, i.e.
$q_{-k}=q_k$, $k\in \mathbb Z^2_*$, and real-valued.  A cylindrical Wiener process in $H$, given  on a filtered probability space  $(\Omega,\mathcal{F},(\mathcal{F}_t)_{t\ge 0}, \mathbb{P})$, can be written as 
 $$
W(t):=\sum_{k\in\bbZ^2_*}B_k(t)e_k,\quad t\ge0.
$$ 
Let $Q\colon H\to H^r$ be a symmetric, positive-definite, bounded linear operator  given by
\begin{equation}
\label{031002}
\widehat {Qw}_k := q_k\widehat w_k,\qquad k\in
\Z^2_*.
\end{equation}
The Hilbert--Schmidt norm of the operator, see Appendix C of \cite{DaPrato-Zabczyk}, can be computed from formula
\begin{equation}
\label{011002}
\| Q\| ^2_{L_{(HS)}(H,H^{r})}:= \sum\limits_{k\in\Z^2_*}\|
Qe_k\| ^2_{H^{r}}= \sum\limits_{k\in\Z^d}|k|^{2r}
q_k^2,
\end{equation}
\begin{proposition}
If $\| Q\| ^2_{L_{(HS)}(H,H^{r})}<+\infty$ then the process $\left(QW(t)\right)_{t\ge0}$ has realizations in $H^r$, $\mathbb{P}$-a.s. Moreover,
the laws of the Wiener processes $\left(\tau_xQW(t)\right)_{t\ge0}$ are independent of $x\in\bbR^2$.
\end{proposition}
\proof The first part of the proposition follows directly from  Proposition 4.2, p. 88 of \cite{DaPrato-Zabczyk}. The second part is a simple consequence of the fact that the processes in question have the same covariance operator as $\left(QW(t)\right)_{t\ge0}$. \qed

\section{Formulation of the main results}\label{sec2.3}

In this section we make it precise what we mean by a solution of \eqref{E11b} with vector field $\vec u$  given by   the solution of the  Navier--Stokes equations \eqref{E11a} and formulate precisely the main results of the paper dealing with the long time, large scale behavior of the trajectory.

Since, as it turns out, the components of the solution of the N.S.E.  belong to $V$, see \cite{MS}, if the initial condition $\vec u_0\in V$, we cannot use equation \eqref{E11b} for a  direct definition of the solution because the point evaluation for the field is not well defined (not to mention the question of the existence and uniqueness of solutions to the o.d.e. in question).

\subsection{Vorticity formulation of the N.S.E}
Note that the rotation 
$$
\xi(t):={\rm rot}\, \vec u(t)=\partial_2u_1(t)-\partial_1u_2(t)
$$ 
of $\vec u(t,x)=(\vec u_1(t,x),\vec u_2(t,x))$,  satisfies
\begin{equation}\label{E25a0}
d \xi(t) =[\Delta\xi(t) -B_0(\xi(t))]d t + Qd  W(t), \qquad \xi(0)=w\in H, 
\end{equation}
with  a cylindrical Wiener process $W(t)$, $t\ge 0$, on $H$, non-anticipative with respect to the filtration $\{{\cal F}_t,\,t\ge0\}$, a certain Hilbert--Schmidt operator $Q\in L_{(HS)}(H,H)$,  and $B_0(\xi):=B_0(\xi,\xi)$,   $\xi\in V$, where
$
B_0(h,\xi):=\vec u\cdot\nabla \xi,
$
with $\vec u:={\cal K}(h)$. 
Let ${\cal E}_T:=C([0,T];H)\cap L^2([0,T];V)$.
\begin{definition}\label{D21a}
{\rm A measurable and $({\cal F}_{t})$-adapted, $H$-valued process $\xi=\left\{\xi(t),\,t\ge0\right\}$ is a solution to $\REQ{E25a0}$  if  for any $T\in (0,+\infty)$, $\xi\in L^2(\Omega, {\cal E}_T, \mathbb{P})$ and 
\begin{equation}
\label{020512a}
\xi(t) = \e ^{\Delta t} w - \int_0^t \e ^{\Delta (t-s)} B_0(\xi(s))d s +   \int_0^t \e ^{\Delta (t-s)} Qd W(s)
\end{equation}
for all $t\ge 0$. }
\end{definition}
The following estimate comes from \cite{MP}, see Lemma A. 3, p. 39.
\begin{proposition}
\label{propMP}
For any $T,N>0$ there exists $C>0$ such that
\begin{equation}
\label{012703}
\bbE\left[\sup_{t\in[0,T]}(|\xi(t)|^2+t\|\xi(t)\|^2)^N\right]\le C(1+|w|^{4N}),\quad \forall\,w\in H.
\end{equation}
\end{proposition}

Let  $\vec u(t):={\cal K}(\xi(t))$. 
Using the above proposition and \eqref{embed} we conclude that
\begin{corollary}
\label{cor012703}
For any $t>0$,  $\vec u(t)\in C(\bbT)$ and
\begin{equation}
\label{022703}
\int_0^t\|\vec u(s)\|_{\infty}ds<+\infty,\quad \mathbb{P}-\mbox{a.s.}
\end{equation}
\end{corollary}
\proof
The continuity of $\vec u(t,x)$ with respect to $x$, follows from the Sobolev embedding. From \eqref{010512a} we conclude that  there exists $C>0$ such that
\begin{equation}
\label{022803}
\|\vec u(s)\|_{\infty}\le C\|\xi(s)\|,\quad\forall\,s\ge 0.
\end{equation}
On the other hand from \eqref{012703} we conclude that for any $t>0$ there exists a random variable   $\tilde C$ that  is almost surely finite
and such that
$
\|\xi(s)\|\le \tilde Cs^{-1/2}
$ for all $s\in(0,t]$.
Combining this with \eqref{022803} we conclude \eqref{022703}.
\qed

\subsection{Definition of trajectory process and its ergodic properties}

\begin{definition}

\label{def-1a}
{\rm Let $x_0\in\bbR^2$. By a  {\em solution to} $\eqref{E11b}$  we mean any $({\cal F}_t)$-adapted  process $x(t)$, $t\ge0$, with   continuous trajectories,  such that
\begin{equation}
\label{integral-ode}
x (t) =x_0+\int_0^t\vec u(s,x(s))d s, \quad\forall\,t\ge0,\qquad \text{$\Bbb P$-a.s.}
\end{equation}
}
\end{definition}

For a given $\nu>0$ denote $e_{\nu}(w):=\exp\{\nu|w|^2\}$, $w\in H$.
\begin{theorem}\label{lab3}
Assume that $Q$ in \eqref{E25a} belongs to $L_{(HS)}(H,V)$  and has a trivial null space, i.e. $Qw=0$ implies $w=0$. 
Suppose that the initial vorticity is random, distributed on $H$ according to the law $\mu_0$ for which
 \begin{equation}
 \label{022011}
 \int_{H}e_{\nu_0}(w)\mu_0(dw)<+\infty
 \end{equation}  
 with a certain $\nu_0>0$.
 Finally, assume that $\{x(t;x_0),\,t\ge0\}$ is a solution of \eqref{E11b} corresponding to the initial data $x_0\in\bbR^2$.
  Then, the following are true:
\begin{enumerate}
\item[1)] \label{1} (Weak law of large numbers)   there exists $v_*=(v_{*,1},v_{*,2})\in\bbR^2$ such that
\begin{equation}
\lim_{T\to+\infty}\frac{x(T;x_0)}{T}= v_{*} \label{spwl}
\end{equation}
in probability.
\item[2)] (Existence of the asymptotic variance) there exists $D_{i j}\in[0,+\infty)$ such that
\begin{equation}
\lim_{T\to+\infty}\frac1T\bbE\left[(x_i(T;x_0)-v_{*,i}T)(x_j(T;x_0)-v_{*,j}T)\right]=D_{ij},\quad i,j=1,2.\label{D}
\end{equation}
\item[3)] \label{2} (Central limit theorem) 
Random vectors $(x(T;x_0)-v_{*}T)/\sqrt{T}$ converge in law, as $T\to +\infty$,  to  a zero mean normal law whose co-variance matrix equals ${\bf D}=[D_{ij}]$.
\end{enumerate}
\end{theorem}

\section{Lagrangian and tracer trajectory processes}


\subsection{Uniqueness in law of the trajectory process}

Define the {\em Lagrangian velocity process} as
$$
\vec \eta(t,x)=(\eta_1(t,x),\eta_2(t,x)):=\vec u(t,x(t)+x),\qquad t\ge 0,\ x\in \mathbb{R}^2. 
$$ 
Suppose  that the forcing $\vec F$  is a white noise in time and spatially homogeneous Gaussian  random field.  
%
Using It\^o's formula  we obtain that its vorticity, given by,
$$
\om(t,x):={\rm rot}\,
\vec\eta(t,x)=\xi(t,x(t)+x)
$$ 
satisfies $\om(0)=\tau_{x_0}w\in H$ and 
\begin{equation}\label{E25a}
d \om(t) =[\Delta\om(t) -B_0(\om(t))+B_1(\om(t))]d t + Qd  W(t),
\end{equation}
where $W$ is   an  $({\cal F}_t)$-adapted cylindrical Wiener process on $H$,  $Q\in L_{(HS)}(H,H)$ and 
$$
B_0(\om):=B_0(\om,\om),  \quad B_1(\om):=B_1(\om,\om),
$$
$$
B_0(h,\om):=\vec \eta\cdot\nabla \om,\quad B_1(h,\om):=\vec \eta(0)\cdot\nabla \om, \quad\om\in V,
$$
with $\vec\eta:={\cal K}(h)$, for more details see \cite{fkp, kp}.  Since we have assumed that $\om\in V$ and, by the Sobolev embedding, ${\cal K}(V)$ is embedded into  the space  $C(\bbT;\bbR^2)$ of two dimensional, continuous trajectory vector fields on $\bbT$, we see that  the evaluation of $\vec \eta$ is well defined, and therefore there is no ambiguity in the definition of 
$B_1(\om)$ for $\om\in V$.

%

\begin{definition}\label{D21}
{\rm A measurable,  $({\cal F}_{t})$-adapted, $H$-valued process $\om=\left\{\om(t),\,t\ge0\right\}$ is a solution to $\REQ{E25a}$,  with the initial condition $\om(0)=w$, if  for any $T>0$, $\om\in L^2(\Omega, {\cal E}_T,\mathbb{P})$ and 
\begin{equation}
\label{020512}
\om(t) = \e ^{\Delta t} w - \int_0^t \e ^{\Delta (t-s)} B_0(\om(s))d s +
 \int_0^t \e ^{\Delta (t-s)}  B_1(\om(s))d s+   \int_0^t \e ^{\Delta (t-s)} Qd W(s),
\end{equation}
$\mathbb{P}$-a.s. for all $t\ge 0$.}
\end{definition}

Sometimes, when we wish to highlight the dependence on the initial condition and the Wiener process, we shall write $\om(t;w,W)$. We shall omit writing one, or both of these parameters when they are obvious from the context.

Using a Galerkin approximation argument, as in Section 3 of  \cite{MS}, see also Appendix \ref{secAp1} below for the outline of the argument, we conclude the following.
\begin{theorem}
\label{MF1}
Given an  initial condition $w\in H$ and an $({\cal F}_t)$-adapted  cylindrical Wiener process $(W(t))_{t\ge0}$,  there exists a unique solution to \eqref{E25a} in the sense of Definition \ref{D21}. Moreover, processes $\{\om(t;w),\,t\ge0\}$ form a Markov family with the corresponding  transition probability semigroup $\{P_t,\,t\ge0\}$ defined on  the space $ C_b(H)$  of continuous and bounded  functions on $H$. 
\end{theorem}

Using the Yamada--Watanabe result, see e.g. \cite{YW} (Corollary after Theorem 4.1.1), or \cite{IW}, from the above theorem we can conclude the following result, see \cite{kp}. 
\begin{corollary}
\label{unique-law}
 Solutions of  \eqref{E25a} have  the uniqueness in law property, i.e.  the laws  over $C([0,+\infty);H)$ of any two solutions of   \eqref{E25a}  starting with the same initial data (but possibly based on different cylindrical Wiener processes) coincide.
\end{corollary}
This immediately implies the uniqueness in law property for solutions of \eqref{E11b}.
\begin{corollary}
\label{unique-law-traj}
 Suppose that $\xi$ and $\xi'$ are two solutions of \eqref{E25a0} with the identical initial data but possibly based on two cylindrical Wiener processes with the respective filtrations $({\cal F}_t)$ and $({\cal F}_t')$. Assume also that $x(\cdot)$ and $x'(\cdot)$ are the solutions of \eqref{E11b} corresponding to $\vec u(t)={\cal K}(\xi(t))$ and $\vec u'(t)={\cal K}(\xi'(t))$,  respectively. Then, the    laws  of the pairs $(x(\cdot),\xi(\cdot))$ and $(x'(\cdot),\xi'(\cdot))$  over $C([0,+\infty), \mathbb R^2)\times C([0,+\infty),H)$  coincide.
\end{corollary}
\proof
Both
$
\om(t,\cdot)=\xi(t,x(t)+\cdot)
$ and $
\om'(t,\cdot)=\xi'(t,x'(t)+\cdot)
$ satisfy \eqref{E25a}. According to Corollary \ref{unique-law} they have identical laws on $C([0,+\infty),H)$ with the initial condition $\tau_{x_0}w$. In fact, due to an analogue of Proposition \ref{propMP} that holds for the process $\om(\cdot)$, see part 1) of Theorem \ref{T2} this law is actually supported in $L^1_{{\rm loc}}([0,+\infty),V)$. We can write therefore that $(x(\cdot),\xi(\cdot))=\Psi(\om(\cdot))$ and $(x'(\cdot),\xi'(\cdot))=\Psi(\om'(\cdot))$, where the mapping
$$
\Psi=(\Psi_1,\Psi_2):L^1_{{\rm loc}}([0,+\infty),V)\to C([0,+\infty), \mathbb R^2)\times C([0,+\infty),H)
$$
is defined as
\begin{eqnarray*}
&&\Psi_1(X)(t):=x_0+\int_0^t{\cal K}(X(s))(0)ds,\\
 && \Psi_2(X)(t,x):=X(t,x-\Psi_1(X)(t)),\quad\forall X\in L^1_{{\rm loc}}([0,+\infty),V), 
 \end{eqnarray*}
and the uniqueness claim made in the corollary follows.
\qed

\subsection{Existence of solution of \eqref{E11b}}

\begin{definition}

\label{def-1}
{\em 
Suppose that   $(\Om,{\cal F}, ({\cal F}_t),\bbP)$ is a filtered probability
space. Let $x_0\in\bbR^2$. By {\em a  weak
solution to} $\eqref{E11b}$  we mean a pair consisting of  a  continuous trajectory $({\cal F}_t)$-adapted  process $x(t)$, $t\ge0$, and an $({\cal F}_t)$-adapted solution $\xi(t)$, $t\ge0$,   to \eqref{E25a0} such that
 \eqref{integral-ode} holds.}
\end{definition}


Suppose now that we are given a filtration $({\cal F}_t)$ and  an  ${\cal F}_t$-adapted solution $\om$ of \eqref{E25a}  with the initial condition $\om(0)=\tau_{x_0}w$.  Define
$
(x(\cdot),\xi(\cdot)):=\Psi(\om(\cdot)).
$
One can easily check, using It\^o's formula, that $(x(\cdot),\xi(\cdot))$  is a  weak solution in the sense of Definition \ref{def-1}. Therefore we conclude the following.
\begin{proposition}
\label{CTM}
Given a filtered probability space there exists a weak solution of  \eqref{E11b}.
\end{proposition}



\section{Spectral gap and regularity properties of the transition semigroup}


\label{sec3}

Here we present the basic  results that shall be instrumental in the proof of Theorem \ref{lab3} formulated in the previous section. In case of the Navier--Stokes dynamics on a two-dimensional torus,  corresponding results have been shown in \cite{HM1}, see Theorem 5.10, Proposition 5.12 and parts 2, 3 of  Lemma A.1  from \cite{HM1}. The proofs of analogous results for the Lagrangian dynamics are not much different, some additional care is needed due to the presence of function $B_1(\cdot)$, but it usually does not create much trouble. We present the proofs of these results in Section \ref{sec6} of the appendix. 

Let us introduce the space  $C_0^\infty(H)$ consisting of all functionals $\phi$, for which there exist $n\ge1$,   a function $F$ from $C^\infty_0(\bbR^n)$  and vectors $v_1,\ldots,v_n\in H$ such that 
$$
\phi(v)=F\left(\langle v,v_1\rangle,\ldots,\langle v,v_n\rangle\right),\quad \forall\,v\in H.
$$
Given $\nu>0$ define ${\cal B}_\nu$ as the completion of $C^\infty_0(H)$ under the norm
$$
\|\phi\|_\nu:=\sup_{w\in H}e_{-\nu}(w)\left(|\phi(w)|+\|D\phi(w)\|\right),
$$
where, as we recall,
$
e_{\nu}(v)=\exp\left\{\nu|w|^2\right\}.
$
Here $\|D\phi(w)\|=\sup_{|\xi|\le 1}|D\phi(w)[\xi]|$, where  $D\phi(w)[\xi]$ denotes the Fr\'echet derivative of a function $\phi\colon H\to\bbR$ at $w$ in the direction $\xi\in H$.
By $\tilde {\cal B}_\nu$ we understand the Banach space of all Fr\'echet  differentiable functions $\phi$ such that
$
\|\phi\|_{\nu}<+\infty.
$
Let $\mathcal P(H)$ be  the space of all Borel, probability measures on $H$. 
Recall also  that $\mu_* \in \mathcal P(H)$ is called an \emph{invariant measure} for $(P_t)_{t\ge 0}$ if
 $$
 \langle \mu_*,P_t\phi\rangle = \langle \mu_*,\phi\rangle, \quad \forall  \,\phi\in C_b(H),\,t\ge 0. 
 $$
 Here  $\langle \mu,\phi\rangle:=\int_H\phi d\mu$ for any   $\mu\in{\cal P}(H)$  and $\phi$ that is integrable.
Our first result can be stated as follows.
\begin{theorem}\label{T1}
Under the assumptions of Theorem \ref{lab3} the following are true:
\begin{itemize} 
\item[1)] 
 there exist $\nu_0, C>0$ such that for any $\nu\in(0,\nu_0]$ we have
\begin{equation}
\label{012011}
\bbE e_{\nu}(\om(t;w))\le Ce_{\nu}(w),\quad\forall\,t\ge0,\,w\in H.
\end{equation}
\item[2)] the constant $\nu_0$ can be further adjusted in such a way that  for any $\nu\in(0,\nu_0]$ the semigroup $(P_t)$ extends to  $\tilde{\cal B}_\nu$ and  
$$
P_t({\cal B}_\nu)\subset {\cal B}_\nu, \quad \forall\,t\ge0.
$$
 In addition, for any $\nu$ as above 
there exist $C,\gamma>0$ such that 
\begin{equation}
\label{010811b}
\|P_t\phi-\langle\mu_*, \phi\rangle\|_{\nu}\le C\e^{-\gamma t}\|\phi\|_{\nu},\quad \forall\,t\ge0,\,\phi\in \tilde {\cal B}_\nu,
\end{equation}
\item[3)] there exist a unique Borel probability measure $\mu_*$ that is  invariant for $(P_t)$,
and  such that 
\begin{equation}
\label{010811a}
\int_He_{\nu}(w)\mu_*(dw)<+\infty,\quad \forall\,\nu\in(0,\nu_0].
\end{equation}
\end{itemize}
\end{theorem}

The property described in \eqref{010811b} is referred to as {\em the spectral gap} of the transition semigroup. Since we shall  use an extension of this property to functions defined on a smaller space than $H$ we introduce the following definition.
For $N>0$ and $\phi\in C^1(V)$ define
$$
\|\!|\phi\|\!|_N:=\sup_{w\in V}\frac{|\phi(w)|+\|D\phi(w)\|}{(1+\|w\|)^N}
$$
and denote by $C^1_N(V)$ the space made of functions, for which $\|\!|\phi\|\!|_N<+\infty$.
\begin{theorem}\label{T2}
Under the assumptions of Theorem \ref{lab3} the following are true:
\begin{itemize}
\item[1)]   for any $t, N>0$ there exists $C_{t,N}$ such that 
\begin{equation} \label{E27-aa}
\bbE\|\om(t;w)\|^N\le C_{t,N}\left(|w|^{2N}+1\right),\quad\forall\,w\in H,
\end{equation}
\item[2)] the definition of the transition semigroup can be extended to an arbitrary   $\phi \in C^1_N(V)$ by letting $P_t\phi(w):=\bbE\tilde\phi(\om(t;w))$, where $\tilde \phi$ is an arbitrary, measurable extension of $\phi$ from $V$ to $H$. Moreover,
 for any $t, N>0$ there exists $C_{t,N}$ such that for any $\nu>0$,  
\begin{equation} \label{E27a}
\|P_t\phi\|_{\nu}\le C_{t,N}\|\!|\phi\|\!|_N,\quad\forall\,\phi\in C^1_N(V).
\end{equation}
\end{itemize}
\end{theorem}
Combining  the above result with part 2) of Theorem \ref{T1} we conclude that the following holds.
\begin{corollary}
\label{cor011011}
For any $N>0$ there exist $C,\nu_0,\gamma>0$ such that for any $\nu\in (0,\nu_0]$  we have
\begin{equation}
\label{010811}
\|P_t\phi-\langle\mu_*, \phi\rangle\|_{\nu}\le C\e^{-\gamma t}\|\!|\phi\|\!|_N,\quad \forall\,t\ge0,\,\phi\in C^1_N(V).
\end{equation}
\end{corollary}

Define 
$$
{\frak p}(w):=\left\{
\begin{array}{ll}
\|w\|^2&\mbox{ for }w\in V,\\
+\infty&\mbox{ for }w\in H\setminus V.
\end{array}
\right.
$$ 
\begin{corollary}
\label{cor011112}
For any $N>0$ we have  $\langle \mu_*,{\frak p}^N\rangle<+\infty$. Thus, in particular $\mu_*(V)=1$.
\end{corollary}
\proof
Suppose that $\varphi_R\colon [0,+\infty)\to[0,R+1]$ is  a continuous function such that $\varphi_R(u)=u$ if $u\in[0,R]$ and it vanishes on $u\ge R+1$. For a fixed $K>0$ we denote
$$
{\frak p}_K(w):=\sum_{0<|k|\le K}|k|^2|\hat w(k)|^2.
$$
Thanks to part  2) of Theorem  \ref{T1} we have $P_t{\frak p}^N\in {\cal B}_\nu$ for any $t>0$ and therefore from \eqref{E27-aa} and \eqref{010811a} we get
\begin{equation}\label{071801}
 \langle \mu_*,P_t{\frak p}_K^N\rangle\le  \langle \mu_*,P_t{\frak p}^N\rangle<+\infty.
\end{equation}
We have therefore
\begin{equation}\label{071801a}
 \langle \mu_*,P_t\varphi_R\circ{\frak p}_K^N\rangle= \langle \mu_*,\varphi_R\circ{\frak p}_K^N\rangle\le  \langle \mu_*,P_t{\frak p}^N\rangle.
\end{equation}
 The first equality follows from the fact that $\mu_*$ is invariant. 
 Letting first $K\to+\infty$ and then subsequently $R\to+\infty$ we conclude the corollary.
\qed


%
%
%
%
%
%
%
%
%
%
%
%
%
%
%
%

\section{Proof of Theorem \ref{lab3}}\label{S6}

\label{sec5}
To abbreviate we assume that  $x_0=0$ and we drop it from our notation. Let $\psi_*=(\psi_*^{(1)},\psi_*^{(2)})\colon V\to\bbR^2$ be defined as  $\psi_*(\om):={\cal K}(\om)(0)$.
Since, for any $s>0$, $H_{1+s}$ is embedded into  $C(\bbT)$,  for any $s>0$ there exists $C>0$ such that
\begin{equation}
\label{H-s}
|\psi_*^{(i)}(w)|\le  C|{\cal K}_i(w)|_{1+s}\le C|w|_s,\quad\forall\,w\in H_s,\,i=1,2.
\end{equation}
It is clear therefore that  the components of $\psi_*$ are bounded linear functional on $V$ and 
 $\psi_*\in C^1_1(V)$.
Suppose also that  $\om(t)$ is the solution of  \eqref{E25} with the initial data distributed according to $\mu_0$.  

\subsection{Proof of part 1)}
\label{sec5.1}

Let $v_*:=(v_{*,1},v_{*,2})$ and 
$
v_{*,i}:=\langle \mu_*,\psi_*^{(i)}\rangle,
$
and $\tilde \psi_*:=\psi_*-v_*$.
To prove the weak law of large numbers it  suffices only to show that for $i=1,2$,
 \begin{equation}
 \label{h11}
\lim_{T\to+\infty} \frac1T\mathbb{E} \tilde x_i(T)
=v_{*,i}\quad\mbox{and}\quad \lim_{T\to+\infty} 
 \frac{1}{T^2}\mathbb{E}\tilde x^2_i(T)= v_{*,i}^{2},
\end{equation}
where 
$$
\tilde x(T)=(\tilde x_1(T),\ldots,\tilde x_d(T)):=\int_{0}^{T} \tilde\psi_*(\om(s))ds.
$$
Using the Markov property we can write that
\begin{eqnarray}
\label{011712}
&&\frac{1}{T}\mathbb{E}\tilde x_i(T)=\frac{1}{T}\int_{0}^{T}
\langle\mu_{0},P_{s}\tilde \psi_{*}^{(i)}\rangle ds,\quad i=1,2.
\end{eqnarray}
Suppose that  $\nu_0$ is chosen in such a way that the conclusions of Theorem \ref{T1} and Corollary  \ref{cor011011} hold. Assume also that 
$\nu\in(0,\nu_0]$. We shall  adjust its value later on.
By virtue of \eqref{010811} we conclude that there exists a constant $C>0$ such that
\begin{equation}
\label{032011}
|P_{t}\tilde \psi_*(w)|\le C\e^{-\gamma t}e_{\nu}(w)\|\!|\tilde \psi_*\|\!|_{1}.
\end{equation}
Hence, the 
 right hand side of \eqref{011712} converges to $0$, by estimate \eqref{022011} and the Lebesgue dominated convergence theorem.
On the other hand
\begin{equation} \label{lab11}
\begin{aligned}
\frac{1}{T^2}\mathbb{E}\tilde x^2_i(T)&=
\frac{1}{T^{2}}\,\mathbb{E}\Big(\int_{0}^{T}\!\!\tilde\psi_{*,i}(\om(t))dt\int_{0}^{T}\!\!\tilde\psi_{*,i}(\om(s))ds\Big) \\
&=\frac{2}{T^{2}}\!\!\int_{0}^{T}\!\!\int_{0}^{t}\!\,\mathbb{E}[\tilde\psi_{*,i}(\om(t))\tilde\psi_{*,i}(\om(s))]dt ds.
\end{aligned}
\end{equation}
The utmost right hand side of   \eqref{lab11} equals
 \begin{eqnarray}
 \label{042011}
&& \frac{2}{T^{2}}\int_{0}^{T}\!\!\int_{0}^{t}\!\,\mathbb{E}\big[\tilde\psi_{*,i}(\om(s))P_{t-s}\tilde\psi_{*,i}(\om(s))\big] dtds=\frac{2}{T^{2}}\int_{0}^{T}\!\!\int_{0}^{t}\! \langle\mu_0P_{s},\tilde\psi_{*,i} P_{t-s}\tilde\psi_{*,i}\rangle dt ds.
 \end{eqnarray}
Using \eqref{032011} we can estimate the right hand side of \eqref{042011} by
\begin{equation}
\frac{C}{T^{2}}\int_{0}^{T}\!\!\int_{0}^{t}\e^{-\ga (t-s)} \langle\mu_0P_{s},|\tilde\psi_{*,i}| e_{\nu}\rangle dt ds=
\frac{C(1-\e^{-\ga T})}{\ga T^{2}}\int_{0}^{T}\langle\mu_0P_{s},|\tilde\psi_{*,i}| e_{\nu}\rangle ds. 
\label{lab4}
\end{equation}
Applying H\"older's inequality with $q\in(1,\nu_0/\nu)$ and an even integer  $p$ such that $p^{-1}:=1-q^{-1}$,  we conclude that  the right hand side is smaller than
 \begin{eqnarray}
 \label{012211}
&&\frac{C}{\ga T^{2}}\int_{0}^{T}\langle\mu_0,P_{s}|\tilde\psi_*|^p \rangle^{1/p} \langle\mu_0P_{s},e_{q\nu}\rangle^{1/q}ds\le \frac{C_1}{\ga T^{2}}\int_{0}^{T}\langle\mu_0,P_{s}|\tilde\psi_*|^p \rangle^{1/p} ds
 \end{eqnarray}
 for some constants $C,C_1$ independent of $T$.
The last inequality follows from \eqref{012011} and \eqref{010811a}. Since $|\tilde\psi_*|^p$ belongs to $C^1_p(V)$ we conclude from Corollaries \ref{cor011011},  \ref{cor011112} and condition \eqref{022011} that the right hand side of the above expression can be estimated by
$ C_2T/(\ga T^{2})$, with $C_2$ a constant independent of $T$, which tends to $0$,
 as $T\to+\infty$. Thus, part 1) follows.
\qed

\subsection{Definition and basic properties of the corrector}  
\label{cor-f}
We start with the following.
\begin{proposition}\label{lab21}
Functions 
\begin{equation}
\label{chi-t}
\chi_{t}(w)=(\chi_{t}^{(1)}(w),\chi_{t}^{(2)}(w)):=\int_{0}^{t}P_s\tilde\psi_*(w) ds,\quad w\in H,
\end{equation}
  converge uniformly on bounded sets, as $t\rightarrow\infty.$ For any $\nu\in(0,\nu_0]$ there is $C>0$ such that
  \begin{equation}
\label{022211}
|\chi_{t}^{(i)}|\le Ce_{\nu},\quad\forall\,t\ge1,\,i=1,2.
\end{equation}
The limit  
\begin{equation}
\label{chi}
\chi=(\chi^{(1)},\chi^{(2)}):=\lim_{t\to+\infty}\chi_{t}=\int_{0}^{+\infty}P_s\tilde\psi_*\,ds,
\end{equation} called a {\em corrector}, satisfies
\begin{equation}
\label{032211}
|\chi^{(i)}|\le Ce_{\nu},\quad i=1,2,
\end{equation}
with the same constant as in \eqref{022211}.
\end{proposition}
\proof
As a  consequence of Corollary \ref{cor011011} we conclude that the functions
$$
 \int_{1}^{t}P_s\tilde\psi_*^{(i)}(w) ds,\quad t\ge1,\,i=1,2,
$$
are well defined on $H$ and converge uniformly on bounded sets. The convergence part of the proposition follows from the fact that there exists a constnt $C>0$ such that for $\nu\in(0,\nu_0]$,
\begin{equation}
\label{022002}
\int_0^1\bbE\|\om(s,w)\|^2ds\le Ce_{\nu}(w),\quad\forall\,w\in H,
\end{equation}
see \eqref{010712} below. This estimate together with \eqref{032011}  imply both \eqref{022211} and \eqref{032211}.
\qed

\begin{proposition}\label{lab21a}
One can choose $\nu_0>0$ in such a way that 
$
\chi^{(i)}\in {\cal B}_{\nu}$ for any $\nu\in(0,\nu_0]$, $i=1,2$.
\end{proposition}
\proof
Since $\tilde\psi_*^{(i)}\in C^1_1(V)$, $i=1,2$, from Corollary \ref{cor011011} we conclude that $P_t\tilde\psi_*^{(i)}\in {\cal B}_{\nu}$ for $t\ge1$ and there exists $\nu_0>0$ such that for any $\nu\in(0,\nu_0]$ one can find $C,\ga>0$, for which
$$
\|P_t\tilde\psi_*^{(i)}\|_{\nu}\le C\e^{-\ga t}\|\!|\tilde\psi_*^{(i)}\|\!|_{1},\quad\forall\, t\ge1,\,i=1,2.
$$
This guarantees that $\int_1^{+\infty}P_t\tilde\psi_*^{(i)}dt$ belongs to ${\cal B}_{\nu}$.  Thanks to estimate \eqref{032211}
it suffices only to show that 
\begin{equation}
\label{031801}
\left|\int_0^{1}DP_t\psi_*^{(i)}(w)[\xi]dt\right|\le Ce_{\nu}(w),\quad\forall\,w,\xi\in H,\,|\xi|\le 1.
\end{equation}
To prove the above estimate note that
$$
\int_0^1 DP_t \psi_*^{(i)}(w)[\xi]dt:=\bbE\left[ {\cal K}(\Xi(1))(0)\right],
$$
where $\Xi(w):= \int_0^1\xi(t;w)dt$ and $\xi(t):=D\om(t;w)[\xi]$. We have, from \eqref{H-s} for $s=1$, that there exists $C>0$ such that
$$
\left|{\cal K}(\Xi(w))(0)\right|\le C\|\Xi(w)\|,\quad\forall\,w\in H.
$$
Hence, from Proposition \ref{prop021801}, we conclude that for any $\nu>0$ there exists $C>0$ such that
$$
\left|\int_0^{1}DP_t\psi_*^{(i)}(w)[\xi]dt\right|^2\le |\xi|^2\bbE\exp\left\{\nu |\om(1)|^2+\frac{\nu}{2e}\int_0^1\| \om(s)\|^2ds\right\}
$$
and \eqref{031801} follows from estimate \eqref{010712} formulated below.
\qed

\subsection{Proof of part 2)}
\label{sec5.3}
After a simple calculation we get
\begin{eqnarray*}
&&D_{ij}(T):=\frac{1}{T}\mathbb{E} \Big[\tilde x_i(T)\tilde x_j(T)\Big]=D_{ij}^1(T)+D_{ij}^2(T),
\end{eqnarray*}
with
\begin{eqnarray*}
&&D_{ij}^1(T):=\frac{1}{T}\int_{0}^{T}\!\! \left\langle \mu_0P_s, \tilde\psi_*^{(i)} \int_{0}^{T-s}P_t\tilde\psi_*^{(j)}\,dt\right\rangle ds,\\
&&
D_{ij}^2(T):=\frac{1}{T}\int_{0}^{T}\!\! \left\langle \mu_0P_s, \tilde\psi_*^{(j)} \int_{0}^{T-s}P_t\tilde\psi_*^{(i)}\,dt\right\rangle ds.
\end{eqnarray*}
It suffices only to deal with the limit of $D_{ij}^1(T)$, the other term can be handled in a similar way. We can write that
\begin{equation*}
\Big|D_{ij}^1(T)-\frac{1}{T}\int_{0}^{T}\! \!\! \left\langle \mu_0 P_{s} ,\tilde\psi_*^{(i)} \chi^{(j)}\right\rangle ds\Big|=\frac{1}{T}\left|\int_{0}^{T}\!\! \left\langle \mu_0P_s, \tilde\psi_*^{(i)}(\chi^{(j)}-\chi_{T-s}^{(j)})\right\rangle ds\right|=R_{ij}(T),
\end{equation*}
where
\begin{equation}
\label{012401}
R_{ij}(T):=\left|\int_{0}^{1}\!\! \left\langle \mu_0P_{sT}, \tilde\psi_*^{(i)}(\chi^{(j)}-\chi^{(j)}_{T(1-s)})\right\rangle ds\right|.
\end{equation}
\begin{lemma}
\label{lm012401}
We have
\begin{equation}
\label{012401c}
\lim_{T\to+\infty}R_{ij}(T)=0.
\end{equation}
\end{lemma}
\proof
Suppose that $p$ is a positive even integer and $q$ is sufficiently close to $1$ so that $q\nu<\nu_0$ and $1/q=1-1/p$, where $\nu$ is as in \eqref{022211} and \eqref{032211}, while $\nu_0$ is such that \eqref{022011} is in force. Then, we can find a  constant $C>0$ such that 
\begin{equation}
\label{022401}
|\chi^{(j)}(w)-\chi^{(j)}_{T(1-s)}(w)|^q\le Ce_{\nu_0}(w),\quad\forall\,w\in H\quad\forall\,s\in[0,1],\,T>0.
\end{equation}
Using Proposition \ref{lab21} and \eqref{022011}  we conclude that
 $$
 \lim_{T\to+\infty} \langle\mu_0P_{sT},|\chi^{(j)}-\chi^{(j)}_{T(1-s)}|^q\rangle=0,\quad\forall\,s\in [0, 1).
$$
Equality \eqref{012401c}  can be concluded, provided we can substantiate passage to the limit with $T$ under the integral appearing on the right hand side of \eqref{012401}.
Suppose first that the argument $s$ appearing in the integral  satisfies $sT\ge 1$.
Using H\"older's inequality, in the same way as it was done in \eqref{012211},   and estimates \eqref{022211} and \eqref{032211} the expression under the integral can be estimated by
\begin{equation}
 \label{052211}
 \begin{aligned}
&
\langle\mu_0,P_{sT}|\tilde\psi_*^{(i)}|^p \rangle^{1/p} \langle\mu_0P_{sT},|\chi^{(j)}-\chi^{(j)}_{T(1-s)}|^q\rangle^{1/q} \\
&\qquad 
\le \sup_{t\ge1}\langle\mu_0,P_{t}|\tilde\psi_*^{(i)}|^p \rangle^{1/p}\langle\mu_0P_{sT},|\chi^{(j)}-\chi^{(j)}_{T(1-s)}|^q\rangle^{1/q}.
 \end{aligned}
 \end{equation}
Since  $|\tilde\psi_*|^p\in C^1_p(V)$ we have $\sup_{t\ge1}\langle\mu_0,P_{t}|\tilde\psi_*|^p \rangle<+\infty$, thanks to part 2) of Theorem \ref{T2}.
As a result the left hand side of \eqref{052211} is bounded for all $s\in[1/ T,1]$.
From the Lebesgue dominated convergence theorem we conclude therefore  that
\begin{equation}
\label{012401a}
\lim_{T\to+\infty}\int_{1/T}^{1}\!\! \left\langle \mu_0P_{sT}, \tilde\psi_*^{(i)}(\chi^{(j)}-\chi^{(j)}_{T(1-s)})\right\rangle ds=0.
\end{equation}

Next we shall prove that there exists $C>0$
such that 
\begin{equation}
\label{052401}
\left|\int_{0}^{1/T}\!\! \left\langle \mu_0P_{sT}, \tilde\psi^{(i)}_*(\chi^{(j)}-\chi^{(j)}_{T(1-s)})\right\rangle ds\right|\le \frac{C}{T},
\end{equation}
provided that $T\ge 1$.
Indeed, using first the Cauchy--Schwartz inequality and then  \eqref{022211}, and  \eqref{032211} we get that the  left hand side can be estimated by
\begin{eqnarray*}
&&
  C\bbE\left\{\left\{  \int_{0}^{1/T}|\tilde\psi_*(\om(sT))|^2ds\right\}^{1/2}\left\{\int_0^{1/T}e_{2\nu}(\om(sT)) ds\right\}^{1/2}\right\}.
\end{eqnarray*}
Applying H\"older's inequality with $q\in(1,2)$ and $1/p=1-1/q$ we get that this expression can be estimated by
\begin{eqnarray*}
&&
C  \left\{\bbE\left\{  \int_{0}^{1/T}|\tilde\psi_*(\om(sT))|^2ds\right\}^{p/2}\right\}^{1/p}\left\{\bbE\left\{\int_0^{1/T}e_{2\nu}(\om(sT)) ds\right\}^{q/2}\right\}^{1/q}\\
  &&
  \le  C_1\left\{\bbE\left\{ \frac1T \int_{0}^{1}\|\om(s))\|^2ds\right\}^{p/2}\right\}^{1/p}\left\{\bbE\left\{\int_0^{1/T}e_{2\nu}(\om(sT)) ds\right\}\right\}^{1/2}
  \\
  &&
  \le 
  \frac{C_2}{T}\left\{\bbE\exp\left\{ \nu \int_{0}^{1}\|\om(s))\|^2ds\right\}\right\}^{p/2}\le \frac{C_3}{T},
\end{eqnarray*}
provided $2\nu<\nu_0$. The penulmative inequality follows from  \eqref{012011} and assumption  \eqref{022011}, while the last estimate is a consequence of \eqref{010712-00} stated below. Thus, \eqref{052401} follows.
 \qed

 We are left therefore with the problem of finding
the limit of
\begin{equation}
\label{011012}
S_{ij}(T)=\frac{1}{T}\int_{0}^{T}\! \!\! \left\langle \mu_0 P_{s} ,\tilde\psi_*^{(i)} \chi^{(j)}\right\rangle ds
\end{equation}
as $T\to+\infty$. Let $R\ge1$ be fixed and $\varphi_R\colon \bbR\to\bbR$ be a smooth mapping such that $\varphi_R(x)=1$ for $|x|\le R$ and  $\varphi_R(x)=0$ for $|x|\ge R+1$.  Observe that 
$$
\hat\chi^{(R)}(w):= \chi^{(j)}(w)\varphi_R(|w|^2)
$$ belongs to $C^1_b(H)$, and thus also to $C^1_b(V)$. Therefore,
$\tilde\psi_*^{(i)}\hat \chi^{(R)}\in C^1_1(V)$. Denote by $S^{(R)}(T)$ the expression in \eqref{011012} with $\chi^{(j)}$ replaced by $\hat\chi^{(R)}$.

Let $\varepsilon>0$ be arbitrary. Using the same argument as in the proof of Lemma \ref{lm012401}  one can show that for any $\varepsilon >0$ there exists a sufficiently large $R\ge 1$ and $T_0>0$ so  that
$$
\left|\frac{1}{T}\int_{0}^{T}\! \!\! \left\langle \mu_0 P_{s} ,\tilde\psi_*^{(i)} (\chi^{(j)}-\hat\chi^{(R)})\right\rangle ds\right|<\frac{\varepsilon}{2}.
$$
Likewise, we can choose $R\ge 1$ and $T_0>0$ so large that
$$
\left|\left\langle \mu_* ,\tilde\psi_*^{(i)} (\chi^{(j)}-\hat\chi^{(R)})\right\rangle \right|<\frac{\varepsilon}{2}.
$$
 By Corollary \ref{cor011011} we have
$$
\|P_t(\tilde\psi_*^{(i)} \hat\chi^{(R)})-\langle\mu_*, \tilde\psi_*^{(i)}\hat \chi^{(R)}\rangle\|_{\nu}\le Ce^{-\gamma t}\|\!|\tilde\psi_* \hat\chi^{(R)}\|\!|_{2},\quad \forall\,t\ge0.
$$
In consequence we conclude that
$$
\lim_{T\to+\infty}S^{(R)}(T)=\langle\mu_*, \tilde\psi_*^{(i)}\hat \chi^{(R)}\rangle.
$$
Hence,
\begin{eqnarray*}
&&\limsup_{T\to+\infty}|S_{ij}(T)-\langle\mu_*, \tilde\psi_*^{(i)} \chi^{(j)}\rangle|\\
&&\qquad 
\le \limsup_{T\to+\infty}|S_{ij}(T)-S^{(R)}(T)|+|\langle\mu_*, \tilde\psi_*^{(i)}\hat \chi^{(R)}\rangle-\langle\mu_*, \tilde\psi_*^{(i)} \chi^{(j)}\rangle|<\varepsilon.
\end{eqnarray*}
This proves that 
$$
\lim_{T\to+\infty}S_{ij}(T)=\langle\mu_*, \tilde\psi_*^{(i)} \chi^{(j)}\rangle.
$$
We have shown 
 therefore  part 2) of the theorem with
 \begin{equation}
 \label{052002}
\lim_{T\to+\infty}D_{ij}(T):=\langle\mu_*, \tilde\psi_*^{(i)} \chi^{(j)}\rangle+\langle\mu_*, \tilde\psi_*^{(j)} \chi^{(i)}\rangle.\qquad \qed
\end{equation}

\subsection{Proof of part 3)} 
\label{sec5.4}
\subsubsection{Reduction to the central limit theorem for  martingales}
\label{mart-decomp}
Note that
\begin{eqnarray}
\label{decomp}
&&\frac{1}{\sqrt{T}}\int_{0}^{T}\tilde\psi_*(\om(s))\,ds
=\frac{1}{\sqrt{T}}M_{T}+R_{T},
\end{eqnarray}
where
\begin{eqnarray}
\label{mart}
M_{T}:=\chi(\om(T))-\chi(\om(0))+\int_{0}^{T}\tilde\psi_*(\om(s))\,ds
\end{eqnarray}
and
$$
R_{T}:= \frac{1}{\sqrt{T}}\left[\chi(\om(0))-\chi(\om(T))\right].
$$
\begin{proposition}\label{lab20}
The process
$\{
M_{T},\,T\ge0\}
$ is a square integrable, two dimensional vector martingale with respect to the  filtration  $\{{\cal F}_T, T\ge 0\}.$ Moreover, random vectors $R_{T}$ converge to $0$, as $T\to+\infty$, in the $L^{1}$-sense.
\end{proposition}
The proof of this result is quite standard and can be found in \cite{kowalczuk}, see Proposition 5.2 and Lemma 5.3.

%
%
%
%
%

\subsubsection{Central limit theorem for martingales}

Assume  that $\{{\cal M}_n,\,n\ge0\}$ is a zero mean martingale subordinated to  a filtration  $\{{\frak F}_{n},\,n\ge0\}$ and 
${\cal Z}_n:={\cal M}_n-{\cal M}_{n-1}$ for $n\ge1$, is the respective sequence of martingale differences. Recall that the quadratic variation of the martingale is defined
as
$$
\langle {\cal M}\rangle_{n}=\sum_{j=1}^{n}\bbE\left[{\cal Z}_j^2|{\frak F}_{j-1}\right],\quad n\ge1.
$$
The following theorem has been shown in  \cite{kowalczuk}, see  Theorem 4.1.
\begin{theorem}\label{lab30b}
 Suppose also that
\begin{itemize}
\item[M1)]
 \begin{equation}
  \label{sublinear}
 \sup_{n\ge1}\bbE {\cal Z}_n^2<+\infty,
 \end{equation}
\item[M2)] \label{lab310} for every $\varepsilon >0,\,$\, 
$$
\lim_{N\to+\infty}\frac {1}{N}\sum_{j=0}^{N-1} \mathbb{E}\Big[ {\cal Z}_{j+1}^2, \,
|{\cal Z}_{j+1} | \ge \varepsilon\sqrt{N}  \Big]=0,
$$  
\item[M3)] 
 there exists $\sigma\ge0$ such that
    \begin{equation}
    \label{010212}
    \lim_{K\rightarrow\infty}\limsup_{\ell\rightarrow\infty}\frac{1}{\ell}\sum_{m=1}^{\ell}\mathbb{E}\Big|\frac{1}{K}
    \bbE\left[\langle {\cal M}\rangle_{mK}-\langle {\cal M}\rangle_{(m-1)K}\Big|{\frak F}_{(m-1)K}\right] -\sigma^{2}\Big|=0,
    \end{equation}
\item[M4)] for every $\varepsilon >0$
\begin{equation}
\label{m3}
\lim_{K\rightarrow\infty}\limsup_{\ell\rightarrow\infty}\frac{1}{\ell K}\sum_{m=1}^{\ell}\sum_{j=(m-1)K}^{mK-1}\mathbb{E}[1+{\cal Z}_{j+1}^2,\, |{\cal M}_{j}-{\cal M}_{(m-1)K}|\geq\varepsilon\sqrt{\ell K}]=0.
\end{equation}
\end{itemize}
Then,
\begin{equation}
\label{M2b}
\lim_{N\to+\infty}\frac{\mathbb{E}\langle {\cal M}\rangle_{N}}{N}
    =\sigma^{2}
\end{equation} and
\begin{equation}
\label{052501}
\lim_{N\to\infty}  \, \mathbb{E}  e^{i \theta
  {\cal M}_N/\sqrt N} = e^{-\sigma^2 \theta^2/2},\quad\forall\,\theta\in\bbR.
\end{equation}
\end{theorem}

\subsubsection{Proof of the central limit theorem for $M_T/\sqrt{T}$}
 
We prove that $M_n/\sqrt{n}$, where $n\ge1$ is an integer, converge in law to a Gaussian random vector, as $n\to+\infty$. This suffices to conclude that
in fact $M_T/\sqrt{T}$ satisfy the central limit theorem.  Indeed,  let   $Z_n:=M_n-M_{n-1}$ for $n
\ge1$. 
Note that  for any $\varepsilon>0$
\begin{equation}
\label{vanish}
\lim_{N\rightarrow\infty}\sup_{T\in[N,N+1)}|M_T/\sqrt{T}-M_N/\sqrt{N}|=0,\quad\bbP-{\rm a.s.}
\end{equation}
For a given $\varepsilon_N>0$ we let
$$
A_N:=[\sup_{T\in[N,N+1)}|M_T/\sqrt{T}-M_N/\sqrt{N}|\geq\varepsilon_N].
$$
We have
\begin{eqnarray*}
\mathbb{P}[A_N]&\le& \mathbb{P}[\sup_{T\in[N,N+1)}|M_T-M_N|\ge \varepsilon_N\sqrt{N}/2]
+\mathbb{P}[|M_N|[N^{-1/2}-(N+1)^{-1/2}]\ge \varepsilon_N/2]\\
&\le& \frac{C}{N^2\varepsilon^4_N}\bbE|Z_{N+1}|^4+\frac{C}{N^3\varepsilon^2_N}\sum_{j=1}^N\bbE|Z_j|^2.
\end{eqnarray*}
The last inequality follows from the Doob and Chebyshev estimates and the  elementary inequality $N^{-1/2}-(N+1)^{-1/2}\le CN^{-3/2}$ that holds for all $N\ge1$ and some constant $C>0$. We denote the first and second  terms on the right hand side by $I_N$ and $I\!I_N$,  respectively. We claim that  there exists $C>0$ such that
\begin{equation}
\label{012311}
\bbE |Z_{N+1}|^4\le C,\quad\forall\,N\ge0.
\end{equation}
Indeed, we have
 \begin{eqnarray*}
&&
\bbE |Z_{N+1}|^4 \le  C\left\{\bbE| \chi(\om(N+1))|^4+\bbE|\chi(\om(N))|^4+\bbE\left|\int_N^{N+1}\tilde\psi_*(\om(s))ds\right|^4\right\}.
\end{eqnarray*}
To estimate the first two terms appearing on the right hand side  we use \eqref{032211} and then subsequently \eqref{010811a}. We conclude that all these terms can be estimated by a constant independent of $N$.  The last expectation can be estimated using \eqref{H-s} by
$$
C\bbE\left[\int_N^{N+1}\|\om(s)\|^2ds\right]^2=C\left\langle \mu_0 P_N,\bbE\left[\int_0^{1}\|\om(s;\cdot)\|^2ds\right]^2\right\rangle.
$$ Applying \eqref{010712} and then again \eqref{010811a}  we obtain that also this term can be estimated independently of $N$.
Hence 
$$
I_N\le \frac{C}{N^2\varepsilon^4_N}.
$$
On the other hand, from \eqref{012311} we conclude also  that for some constants $C,C_1>0$ independent of $N$ we have
$$
I\!I_N= \frac{C}{N^3\varepsilon^2_N}\sum_{k=1}^N\bbE |Z_k|^2\le \frac{C_1}{N^2\varepsilon^2_N}.
$$
Choosing $\varepsilon_N$ tending to $0$ sufficiently slowly we can guarantee that
$$
\sum_{N\ge1}\mathbb{P}[A_N]<+\infty,
$$
and \eqref{vanish} follows from an application of the Borel--Cantelli lemma.

Choose $a\in\bbR^2$ and let
$
{\cal M}_n:=M_n\cdot a.
$
Condition M1) obviously holds in light of \eqref{012311}. Condition M2) also easily follows from \eqref{012311}  and the Chebyshev inequality.
Before verifying hypothesis M3) let us introduce
some additional notation. For a given probability measure $\mu$ on $H$ and a Borel event $A$  write 
$$
\bbP_{\mu}[A]:=\int_H\bbP[A|\om(0)=w]\mu(dw).
$$
The respective expectation shall be denoted by $\bbE_\mu$. We write $\bbP_w$ and  $\bbE_w$ in case of $\mu=\delta_w$.  We can write that
\begin{eqnarray*}
&&
\frac{1}{K} \bbE\left[\langle{\cal M}\rangle_{mK}-\langle {\cal M}\rangle_{(m-1)K}\Big|{\frak F}_{(m-1)K}\right] =\frac{1}{K}\sum_{j=0}^{K-1}P_{j}\Psi(\om((m-1)K))
\end{eqnarray*}
with
$
\Psi(w):=\bbE_w {\cal M}_1^2.
$
Suppose that
$\si^2=\langle \mu_*,\Psi\rangle.$ Let also $
\tilde \Psi(w):=\Psi(w)-\si^2,
$
$$
S_K(w):=\frac{1}{K}
   \sum_{j=0}^{K-1}P_{j}\tilde\Psi(w)
$$
and 
$$
\tilde S_K(w):=|S_K(w)|-\langle \mu_*,|S_K|\rangle,\quad\,w\in H.
$$
We can rewrite the expression under the limit in \eqref{010212} as being equal to
 \begin{equation}
    \label{010212a}
 \frac{1}{\ell}\sum_{m=1}^{\ell}\mathbb{E}\Big|\frac{1}{K}
   \sum_{j=0}^{K-1}P_{j}\tilde\Psi(\om((m-1)K))\Big|= \langle \mu_0Q_\ell^K, |S_K|\rangle, 
    \end{equation}
    where
    $$
Q_\ell^K:=\frac{1}{\ell}\sum_{m=1}^{\ell} P_{(m-1)K}.
$$
It is obvious that the second term on the right hand side of \eqref{010212a} does not contribute to the limit in hypothesis M3).
We prove that
\begin{equation}
    \label{010212b}
\lim_{\ell\to+\infty}\sum_{m=1}^{\ell}\langle \mu_0Q_\ell^K, \tilde S_K\rangle =0.
    \end{equation}
Then M3) shall follow upon subsequent applications  of \eqref{010212b}, as $\ell\to+\infty$, and Birkhoff's individual ergodic theorem, as $K\to+\infty$.
To prove \eqref{010212b} it suffices only to show that the function
$S_K(\cdot)$ is continuous on $H$
and for any $K$ fixed there exists a constant $C>0$ such that
\begin{equation}
\label{030212}
 |S_K(w)|\le Ce_{\nu}(w),\quad \forall\,w\in H.
\end{equation}
Equality \eqref{010212b} is then a consequence of the fact that
measures  
$
\mu_0Q_\ell^K
$ converge weakly to $\mu_*$ as $\ell\to+\infty$, and 
estimate \eqref{012011}.
Continuity of $S_K(\cdot)$ follows from the  fact that $\tilde\Psi\in{\cal B}_{\nu}$. On the other hand estimate \eqref{030212} follows from the fact that
 for any $j\ge 1$ fixed there exists a constant $C>0$ such that
\begin{equation}
\label{030212a}
P_{j}\Psi(w)\le Ce_{\nu}(w),\quad \,w\in H.
\end{equation}
The last estimate can be seen as follows
\begin{eqnarray}
\label{marta}
&&
\Psi(w)\le |a|^2\bbE_w|M_{1}|^2=|a|^2\sum_{i=1}^2\left\{P_1[\chi^{(i)}]^2(w)+[\chi^{(i)}(w)]^2+2\int_{0}^{1}P_s(\tilde\psi_*^{(i)}P_{1-s}\chi^{(i)})(w)\,ds\right.\nonumber\\
&&
\\
&&
\left.+2\int_{0}^{1}ds\int_0^sP_{s'}(\tilde\psi_*^{(i)}P_{s-s'}\tilde\psi_*^{(i)})(w)\,ds'+2(\chi^{(i)}P_1\chi^{(i)})(w)+2\chi^{(i)}(w)\int_{0}^{1}P_s\tilde\psi_*^{(i)}(w)\,ds\right\}.\nonumber
\end{eqnarray}
Using estimates \eqref{012011} and \eqref{032211} we
conclude that for any $\nu>0$ there exists a constant $C>0$ such that
$$
\Psi(w)\le Ce_{\nu}(w),\quad \forall\,w\in H.
$$
Hence, using again \eqref{012011}, we conclude \eqref{030212a}. This ends the proof of hypothesis M3).

Finally we verify condition M4). For that purpose  it suffices only to prove that
\begin{eqnarray*}
 \lim_{K\rightarrow +\infty}\limsup_{\ell\rightarrow +\infty}\frac 1K\sum_{j=0}^{K-1}\langle \mu_{0}Q_{\ell}^K, G_{\ell,j}\rangle=0,
\end{eqnarray*}
where
$$
G_{\ell,j}(w):= \mathbb{E}_w\left[1+|Z_{j+1}|^2,|M_{j}|\geq\varepsilon\sqrt{\ell K}\right].
$$
The latter follows if we show that 
\begin{eqnarray}
\label{G-j}
\limsup_{\ell\rightarrow+\infty}\langle \mu_{0}Q_{\ell}^K, G_{\ell,j}\rangle=0,\quad\forall\,j=0,\ldots,K-1.
\end{eqnarray}
From the Markov inequality we obtain
$$
\mathbb{P}_w\Big[|M_{j}|\geq \varepsilon\sqrt{\ell K}\Big]\leq \frac{\mathbb{E}_{w}|M_{j}|}{\varepsilon\sqrt{\ell K}}
\le I_1+I_2,
$$
where
$$
I_1:=
\frac{1}{\varepsilon\sqrt{\ell K}}\sum_{i=1}^2\mathbb{E}_{w}|\chi^{(i)}(\om(j))-\chi^{(i)}(w)|
$$
and
$$
I_2:=\frac{1}{\varepsilon\sqrt{\ell K}}\sum_{i=1}^2\mathbb{E}_{w}\left|\int_0^j\tilde\psi_*^{(i)}(\om(s))ds\right|.
$$
Using \eqref{032211} 
we conclude that
$$
I_1\leq \frac{C_1e_{\nu}(w)}{\varepsilon\sqrt{\ell K}}.
$$
On the other hand, we have
$$
I_2\le \frac{C_2}{\varepsilon\sqrt{\ell K}}\bbE_w\int_0^j\|\om(s)\|ds
$$
  and from \eqref{010712-00} we get that 
 $$
I_2\leq \frac{C_3e_{\nu}(w)}{\varepsilon\sqrt{\ell K}}.
$$ 
  Summarizing, we have shown that
  for any $R>0$,
\begin{eqnarray}
\label{012201-2011}
\sup_{|w|\le R}\mathbb{P}_{w}\Big[|M_{j}|\geq \varepsilon\sqrt{\ell K}|\Big]\leq \frac{C}{\sqrt{\ell K}}.
\end{eqnarray}
Furthermore,
\begin{eqnarray}
\label{022611}
&&\sup_{|w|\le R}\mathbb{E}_{w}\left[|Z_{j+1}|^2,| M_{j}|\geq\varepsilon\sqrt{\ell K}\right]\\
&&
\leq\,
2\sum_{i=1}^2\left\{\sup_{|w|\le R}\mathbb{E}_{w}\left\{\left[\chi^{(i)}(\om(j+1))-\chi^{(i)}(\om(j))\right]^2,| M_{j}|\geq\varepsilon\sqrt{\ell K}\right\}\right.\nonumber\\
&&+
\left.\sup_{|w|\le R}\mathbb{E}_{w}\left\{\Big[\int^{j+1}_{j}\tilde\psi_*^{(i)} (\om(s))ds\Big]^{2},| M_{j}|\geq\varepsilon\sqrt{\ell K}\right\}\right\}\nonumber
\\
&&
\le C\sup_{t\in[0,K]}\sup_{|w|\le R}\mathbb{E}_{w}\left[e_\nu(\om(t)),| M_{j}|\geq\varepsilon\sqrt{\ell K}\right]\nonumber
\end{eqnarray}
for some constant $C$ independent of $\ell$. The above argument shows that
$$
\lim_{\ell\to+\infty}\sup_{|w|\le R}|G_{\ell,j}(w)|= 0.
$$

To obtain \eqref{G-j} it suffices only to prove  that for  $\delta>0$ as in H3) we have
 \begin{equation}
\label{032201-2011}
 \limsup_{\ell\rightarrow+\infty}\langle \mu_{0}Q_{\ell}^K, G_{\ell,j}^{1+\delta/2} \rangle<+\infty,\quad\forall\,K\ge1, 
\,0\le j\le K-1.
\end{equation} 
Note that 
\begin{eqnarray}
\langle \mu_{0}Q_{\ell}^K, G_{\ell,j}^{1+\delta/2} \rangle\le\mathbb{E}_{\mu_{0}Q_{\ell}^K}(1+|Z_{j+1}|^2 )^{1+\delta/2}.\label{lab48}
\end{eqnarray}
This however is a consequence of \eqref{012011}. Thus condition M4) follows.

Summarizing, we have shown that 
$$
\lim_{n\to+\infty}\exp\left\{\frac{ia\cdot M_N}{\sqrt{N}}\right\}=\exp\left\{-\frac{1}{2}\sum_{i,j=1}^2D_{ij}a_ia_j\right\},
$$
where
$$
D_{ij}:=\left\langle\mu_*,\bbE\left\{\prod_{p=i,j}\left[\chi^{(p)}(\om(1;w))-\chi^{(p)}(w)+\int_{0}^{1}\tilde\psi_*^{(p)}(\om(s;w))\,ds\right]\right\}\right\rangle.
$$
After a somewhat lengthy, but straightforward calculation, using stationarity of $\mu_*$ and the fact that
$$
\left\langle\mu_*,\left[P_s\chi^{(i)}-\chi^{(i)}+\int_{0}^{s}P_{s'}\tilde\psi_*^{(i)}\,ds'\right]\tilde\psi_*^{(j)}\right\rangle=0,\quad\forall\,s\ge0
$$
we conclude that
$D_{ij}$ coincides with the expression on the right hand side of  \eqref{052002}. \qed

\section{Proof of the results from section \ref{sec3}}

\label{sec6}

\subsection{Proof of Theorem \ref{T1}}

Part 3) is a direct consequence of parts 1) and 2). 

\subsubsection{Proof of part 1)}

Suppose that $\om(t):=\om(t;w)$.
From \eqref{010712} to conclude that for $\nu\in(0,\nu_0]$, where $\nu_0=1/(4\|Q\|)$, there exists a constant $C>0$ such that
\begin{equation}
\label{010712a}
\bbE\exp\left\{\nu|\om(n+1)|^2\right\}\le C\bbE\exp\left\{q\nu|\om(n)|^2\right\},\quad\forall\,n\ge 0.
\end{equation}
Let $q=e^{-1/2}$. The right hand side can be further estimated using Jensen's inequality
$$
C\bbE\exp\left\{q\nu|\om(n)|^2\right\}\le C\left(\bbE\exp\left\{\nu|\om(n)|^2\right\}\right)^q\le   C^{1+q}\left(\bbE\exp\left\{q\nu|\om(n-1)|^2\right\}\right)^q.
$$
Iterating this procedure we conclude that for any $n\ge 0$
\begin{equation}
\label{010712b}
\begin{aligned}
\bbE\exp\left\{\nu|\om(n+1)|^2\right\}&\le C^{1+q+\ldots+q^n}\left\{\exp\left\{q^{n+1}\nu|\om(0)|^2\right\}\right\}^{1/q^{n+1}}
\\
&\le C^{1/(1-q)}\exp\left\{\nu|w|^2\right\}.
\end{aligned}
\end{equation}
Therefore (cf. part 3) of  Lemma A.1 of \cite{HM1}) we have the following.
\begin{lemma}
\label{lm011112}
There exists a constant $C>0$ such that 
\begin{equation}
\label{010712c}
\bbE\exp\left\{\nu|\om(t;w)|^2\right\}\le C\exp\left\{\nu|w|^2\right\},\quad\forall\,t\ge0,\,\nu\in(0,\nu_0],\,w\in H.
\end{equation}
\end{lemma}
The above lemma obviously  implies \eqref{012011}.

\subsubsection{A stability result of Hairer and Mattingly}

In our proof we use  Theorems 3.4 and 3.6 of \cite{HM1}, which we recall below. 
Suppose that $({\cal H},|\cdot|)$ is a separable Hilbert space with a stochastic flow $\Phi_t\colon {\cal H}\times\Omega \to{\cal H}$, $t\ge0$, i.e.  a family of $C^1$-class random mappings of ${\cal H}$ defined over a probability space $(\Om,{\cal F},\bbP)$ that satisfies $\Phi_t(\Phi_s(x;\om);\om))=\Phi_{t+s}(x;\om)$ for all $t,s\ge0$, $x\in{\cal H}$ and $\bbP$ a.s. $\om\in
\Om$. 
We assume that $P_t$ and $P_t(x,\cdot)$, $x\in {\cal H}$,  are transition semigroup and a family of  transition probabilities corresponding to the flow, i.e.
$$
P_t \phi(x)=\int \phi(y)P_t(x,dy)=\bbE \phi(\Phi_t(x)),\quad\forall\phi\in B({\cal H}), \,x\in{\cal H}.
$$
Here $B({\cal H})$ is the space of Borel and bounded functions on ${\cal H}$. The dual semigroup acting on a Borel probability measure $\mu$ shall be 
 denoted by $\mu P_t$.
We adopt the following hypotheses on the flow.

{\bf Assumption 1.} There exists a measurable function $V\colon {\cal H}\to[1,+\infty)$ and two increasing continuous functions $V_*,V^*\colon [0,+\infty)\to[1,+\infty)$ that satisfy
\begin{itemize}
\item[1)]
$$
V_*(|x|)\le V(x)\le V^*(|x|),\quad \forall\,x\in{\cal H},
$$
and $\lim_{a\to+\infty}V_*(a)=+\infty$,
\item[2)] there exist $C>0$ and $\kappa_1>1$ such that
$$
aV^*(a)\le CV_*^{\kappa_1}(a),\quad\forall\,a\ge0,
$$
\item[3)] there exist $\kappa_0<1$, $C>0$ and a decreasing function $\alpha\colon [0,1]\to[0,1]$ with $\alpha(1)<1$ such that
$$
\bbE\left[V^{\kappa}(\Phi_t(x))\left(1+|D\Phi_t(x)[h]|\right)\right]\le CV^{\alpha(t)\kappa}(x),\quad\forall\,x,h\in {\cal H},\,|h|=1,
$$
and $t\in[0,1]$, $\kappa\in[\kappa_0,\kappa_1]$.
Here $D\Phi_t(x)[h]$ denotes the Fr\'echet derivative at $x$ in the direction $h$. 
\end{itemize}

{\bf Assumption 2.} There exist $C>0$ and $\kappa_2\in[0,1)$ such that for any $\varepsilon \in(0,1)$ one can find $C(\varepsilon),T(\varepsilon)>0$,
for which
\begin{equation}
\label{E26}
|DP_t\phi(x)|\le CV^{\kappa_2}(x)\left\{C(\varepsilon)\left[P_t(|\phi|^2)(x)\right]^{1/2}+\varepsilon \left[P_t(|D\phi|^2)(x)\right]^{1/2}\right\},
\end{equation}
for all $x\in{\cal H}$, $t\ge T(\varepsilon)$.

Introduce now the following family of metrics on ${\cal H}$. For $\kappa\ge0$ and $x,y\in {\cal H}$ we let
$$
{\rm d}_\kappa(x,y):=\inf_{c\in \Pi(x,y)}\int_0^1V^{\kappa}(c(t))|\dot c(t)|dt,
$$
where the infimum extends over the set $\Pi(x,y)$ consisting of all $C^1$ regular paths $c\colon [0,1]\to{\cal H}$ such that $c(0)=x$, $c(1)=y$. In the special case of $\kappa=1$ we set  ${\rm d}={\rm d}_1$.
For two Borel probability measures $\mu_1,\mu_2$ on ${\cal H}$ denote by ${\cal C}(\mu_1,\mu_2)$ the family of all Borel measures on ${\cal H}\times{\cal H}$ whose marginals on the first and second coordinate coincide with $\mu_1,\mu_2$ respectively. We denote also by
$$
{\rm d}(\mu_1 ,\mu_2 ):=\sup\left[\left|\langle \mu_1,\phi\rangle-\langle \mu_2,\phi\rangle\right|\colon {\rm Lip}(\phi)\le 1\right].
$$
Here ${\rm Lip}(\phi)$ is the Lipschitz constant of $\phi\colon {\cal H}\to\bbR$ in the metric ${\rm d}(\cdot,\cdot)$. By ${\cal P}_1({\cal H},{\rm d})$ we denote the space of all Borel, probability measures $\mu$ on ${\cal H}$ satisfying $\int_{\cal H} {\rm d}(x,0)\mu(dx)<+\infty$.
 
Let  $A\subset {\cal H}\times{\cal H}$ be Borel measurable. For a given $t\ge0$ and $x,y\in{\cal H}$ denote
$$
{\cal P}_t(x,y;A)=\sup\left[\mu[A]\colon \mu\in {\cal C}(P_t(x,\cdot),P_t(y,\cdot))\right].
$$

{\bf Assumption 3.} Given any $\kappa\in(0,1)$ and $\delta,R>0$ there exists $T_0>0$ such that for any $T\ge T_0$ there exists $a>0$ for which
$$
\inf_{|x|,|y|\le R}{\cal P}_T(x,y;\Delta_{\delta,\kappa})\ge a.
$$
Here, 
$$
\Delta_{\delta,\kappa}:=[(x,y)\in {\cal H}\times{\cal H}\colon {\rm d}_\kappa(x,y)<\delta],\quad \forall\,\kappa,\delta>0.
$$

\begin{theorem}
\label{hm-result}
Suppose that Assumptions 1, 2, 3 stated above are in force. Then the following are true:
\begin{itemize}
\item[1)]
 there exist $C,\gamma>0$ such that
\begin{equation}
\label{010812}
{\rm d}(\mu_1 P_t,\mu_2 P_t)\le C\e^{-\gamma t}{\rm d}(\mu_1,\mu_2),\quad \forall \mu_1,\mu_2\in {\cal P}_1({\cal H},{\rm d}),
\end{equation}
\item[2)]
there exists a unique probability measure $\mu_*\in {\cal P}_1({\cal H},{\rm d})$ invariant under $\{P_t,\,t\ge0\}$, i.e. $\mu_*=\mu_* P_t$ for all $t\ge0$,
\item[3)] we have
\begin{equation}
\label{020812}
\| P_t\phi-\langle\mu_*,\phi\rangle\|_{\rm Lip}\le  C\e^{-\gamma t}\| \phi-\langle\mu_*,\phi\rangle\|_{\rm Lip},\quad \forall\,\phi\in C^1({\cal H}),\,t\ge0.
\end{equation}
Here 
$$
\|\phi\|_{\rm Lip}:=\sup_{x\not=y}\frac{|\phi(x)-\phi(y)|}{{\rm d}(x,y)}+|\langle \mu_*,\phi\rangle|.
$$
\end{itemize}
\end{theorem}

\subsubsection{Proof of part 2)}
\subsection*{Verification of Assumption 1} 
Denote  $\Phi_t(w;W):=\om(t;w,W)$, where $W$ is the cylindrical Wiener process appearing in \eqref{E25a}. Let 
\begin{equation}
\label{041912}
\xi(t;w,\xi):=D\Phi_t(w)[\xi], \quad\xi\in H.
\end{equation}
In what follows we suppress $w$ and $\xi$ in our notation when their values are obvious from the context.
Define
$V(w):=V_*(|w|)=V^*(|w|)=\e^{\nu |w|^2}$. Assumption 1 of Theorem \ref{hm-result} is a consequence of the result below and estimate
 \eqref{010712} shown in the Appendix \ref{secApB}.
 \begin{proposition}
 \label{prop021801}
For any $\nu>0$ there exists $C>0$ such that
\begin{equation}
\label{051801}
 |\xi(t)|\le  |\xi|\exp\left\{\nu\int_0^t\| \om(s)\|^2ds+Ct\right\}
 \end{equation}
 and
 $$
\left\{ \int_0^t\|\xi(s)\|^2ds\right\}^{1/2}\le  |\xi|\exp\left\{\nu\int_0^t\| \om(s)\|^2ds+Ct\right\},\quad\forall\,t\ge0,\quad\bbP-{\rm a.s.}
 $$
 \end{proposition}
 \proof
Note that $\xi(t)$ satisfies a  (non-stochastic) equation
\begin{eqnarray}\label{E25-1}
&&
\partial_t \xi(t) =\Delta\xi(t) -\eta(t)\cdot\nabla \xi(t)-{\cal K}(\xi(t))\cdot\nabla \om(t)\\
&&
+\eta(t,0)\cdot\nabla \xi(t)+{\cal K}(\xi(t))(0)\cdot\nabla \om(t), \qquad \xi(0)=\xi\in H. \nonumber
\end{eqnarray}
Hence,
$$
\partial_t |\xi(t)|^2 =-2\|\xi(t)\|^2-2\langle {\cal K}(\xi(t))\cdot\nabla \om(t),\xi(t)\rangle+2\langle{\cal K}(\xi(t))(0)\cdot\nabla \om(t),\xi(t)\rangle.
$$
Using \eqref{020712} and \eqref{030712} (for $r=1/2$) we conclude that for some deterministic $C>0$,
\begin{eqnarray*}
\partial_t |\xi(t)|^2 &&\le -2\|\xi(t)\|^2+C|\xi(t)|_{1/2}\| \om(t)\||\xi(t)|
\\
&&
\le -2\|\xi(t)\|^2+\nu\| \om(t)\|^2|\xi(t)|^2+\frac{C^2}{4\nu}|\xi(t)|_{1/2}^2.
\end{eqnarray*}
An application of  the Gagliardo--Nirenberg inequality  \eqref{gagliardo-nirenberg} with $s=1$, $\beta=1/2$ yields
$$
|\xi(t)|_{1/2}\le C\|\xi(t)\|^{1/2}|\xi(t)|^{1/2}
$$
for some constant $C>0$.
In consequence, there exist $C,C_1>0$ such that
\begin{equation}\label{051801a}
\begin{aligned}
\partial_t |\xi(t)|^2 &\le -\|\xi(t)\|^2+ \nu\| \om(t)\|^2|\xi(t)|^2+\frac{C^2}{2\cdot 4^3\nu}|\xi(t)|^2\\
&
\le -\|\xi(t)\|^2+ (\nu\| \om(t)\|^2+C_1)|\xi(t)|^2.
\end{aligned}
\end{equation}
 Estimate \eqref{051801} follows upon an application of Gronwall's inequality.
 In addition, from  \eqref{051801} and \eqref{051801a}  we conclude that there exists $C>0$ such that
 \begin{eqnarray*}
\int_0^t\|\xi(s)\|^2ds&\le& |\xi|^2+ \int_0^t (\nu\| \om(s)\|^2+C_1)|\xi(s)|^2ds\\
&\le& 
 |\xi|^2+ |\xi|^2\int_0^t (\nu\| \om(s)\|^2+C)\exp\left\{\nu\int_0^s\| \om(u)\|^2du+Cs\right\}ds\\
&
 \le &|\xi|^2\exp\left\{\nu\int_0^t\| \om(s)\|^2dt+Ct\right\}. \qed
\end{eqnarray*}

\subsection{Verification of Assumption 2}

%

Here we follow the ideas of  Hairer and Mattingly, see \cite{HM1}. Suppose that  $\Psi\colon H\to {\cal H}$ is a Borel measurable function. Given an $(\cF_t)$-adapted process $g\colon [0,\infty) \times \Omega \to H$ satisfying $ \E \int_0^t |g_s|^2d s <+\infty $ for each $t\ge0$ we denote by $\cD_g\Psi(\om(t))$  the Malliavin derivative of $\Psi(\om(t))$ in the direction of $g$; that is 
$$
\cD_g\Psi(\om(t;w)):=\lim_{\varepsilon \downarrow
0}\frac{1}{\varepsilon} \left[ \Psi(\om(t;w,W+\varepsilon g)) -\Psi(
\om(t;w,W))\right], 
$$
where the limit is understood  in the $L^2(\Omega,\cF,\P;{\cal H})$ sense. Recall that $\om_g(t;w):=\om(t;w,W+g)$ solves the equation
\begin{equation}\label{E25}
\begin{aligned}
d \om_g(t;w) &=[\Delta\om_g(t) -B_0(\om_g(t;w))+B_1(\om_g(t;w))]d t + Qd  W(t)+Qg(t)dt, 
\\
 \om(0;w)&=w\in H. 
\end{aligned}
\end{equation}
The following two facts about the Malliavin derivative shall be crucial
for us in the sequel. Directly from the definition of the Malliavin
derivative we conclude \emph{the chain rule:} suppose that $\Psi\in
C^1_b(H;{\cal H})$ then
\begin{equation}\label{083002}
\cD_g\Psi( \om(t;w))=D\Psi(\om(t;w))[D(t)],
\end{equation}
with $D(t;w,g)=: \cD_g\om(t;w)$, $t\ge 0$.
In addition, the  \emph{integration by parts formula} holds, see
Lemma 1.2.1, p. 25 of \cite{Nualart}. Suppose that $\Psi\in
C^1_b(H)$ then
\begin{equation}
\label{083003} \E[\cD_g\Psi( \om(t;w))]=
\E\left[\Psi(\om(t;w))\int_0^t\langle g(s),d
W(s)\rangle\right].
 \end{equation}

In particular, one can easily show that when $H={\cal H}$ and $\Psi=I$, where $I$ is the identity operator,
the Malliavin derivative of $\om(t;w)$ exists and
 the process $D(t;w,g)$ (we omit writing $w$ and $g$ when they are obvious from the context), solves the linear equation
\begin{equation}\label{ET28}
\begin{aligned}
\frac{d D}{d t}(t)&=\Delta D(t) -\eta(t)\cdot \nabla D(t)-\delta k(t)\cdot \nabla \om(t)\\
&+\eta(t,0)\cdot \nabla D(t)+\delta k(t,0)\cdot \nabla \om(t)
+ Qg(t),\\
&\\
D(0)&=0.
\end{aligned}
\end{equation}
Here
$\delta k(t):={\cal K}(D(t))$.
Denote $\rho(t;w,\xi):=\xi(t)- {\cD}_g\om(t;w)$. We have the following.
\begin{proposition}
\label{m-lm} For any $\nu,\gamma>0$ there exists a constant $C>0$ such that for any given $w,\xi\in H$ one can find  an $(\cF_t)$-adapted $H$-valued
process $g(t)= g(t,w,\xi)$ that satisfies
\begin{equation}\label{ET210}
\sup_{|\xi|\le 1} \E \,
|\rho(t;w,\xi)|^2\le Ce_{\nu}(w)e^{-\gamma t},\quad\forall\,t\ge0,
\end{equation}
and
\begin{equation}\label{ET29}
\sup_{|\xi|\le 1} \int_0^\infty \E\,
|g(s,w,\xi)|^2d s \le Ce_{\nu}(w),\quad \forall\,w\in H.
\end{equation}
\end{proposition}
We prove this proposition shortly. First, however let us demonstrate
 how to use it to finish  verification of Assumption 2. We have
$$
D P_t\phi(w)[\xi]= \ \  \E\left\{\,D\phi(\om(t;w))[D(t)]\right\}
+\E\,\left\{ D
\phi(\om(t;w))[\rho(t;w,\xi)]\right\}.
$$
Using the chain rule, see \eqref{083002}, the right hand side can be rewritten as
\begin{eqnarray*}
&&
 \E\, \left\{{\cD}_g \phi(\om(t;w))\right\} +\E\,\left\{ D
\phi(\om(t;w))[\rho(t;w,\xi)]\right\}
\\
&&
= \E\,
\left\{\phi(\om(t;w))\int_0^t\langle g(s),\d
W(s)\rangle\right\}+ \E\,\left\{
 D\phi(\om(t;w))[\rho(t;w,\xi)]\right\}.
 \end{eqnarray*}
 The last  equality follows from integration by parts formula  \eqref{083003}.
We have
$$
\left| \E\, \left\{\phi(\om(t;w))\int_0^t\langle g(s),d
W(s)\rangle\right\}\right|\le \left(P_t|\phi|^2(w)\right)^{1/2}\left(\E\,
\int_0^\infty |g(s)|^2d s\right)^{1/2}
$$
and
$$
\left|\E\,\left\{
 D\phi(\om(t;w))[\rho(t;w,\xi)]\right\} \right|\le \left(P_t|D\phi|^2(w)\right)^{1/2}\left(\E\,
|\rho(t;w,\xi)|^2\right)^{1/2}.
$$
Hence, by \REQ{ET29} and \REQ{ET210}, given  $\kappa_2\in(0,1),$ $\nu>0$ , the corresponding $V(w)=e_{\nu}(w)$ and $\varepsilon \in(0,1)$,  we conclude 
estimate \REQ{E26} with $T_0$, $C(\varepsilon)$,  such that
$$
\left(\E\,
\int_0^\infty |g(s)|^2d s\right)^{1/2}\le C(\varepsilon)V^{\kappa_2}(w)$$
and
$$
 \sup_{|\xi|\le 1}\sup_{t\ge T_0} \left\{\E \,
|\rho(t;w,\xi)|^2\right\}^{1/2}\le \varepsilon V^{\kappa_2}(w).
$$
Therefore Assumption 2 will be verified, provided  that
 we prove Proposition \ref{m-lm}.

\subsubsection*{Proof of Proposition \ref{m-lm}}

We assume first that $q_k\not=0$ for all $k\in\bbZ^2_*$, see \eqref{031002}.
 Let us denote by $\Pi_{\ge N}$ the orthogonal projection
onto span $\left\{\e^{\i kx}\colon |k|\ge N\right\}$ and  let
$\Pi_{< N}:=I-\Pi_{\ge N}$. Write 
$$
B:=-B_0+B_1, \quad B_s(h,\om):=B(h,\om)+B(\om,h),
$$ 
$B_{i,s}(\cdot,\cdot)$ for the  symmetrized forms corresponding to $B_i$, $i=0,1$, 
and
$$
\Delta_N:=\Pi_{\ge N}\Delta,\quad  Q_N:=\Pi_{\ge N}B,\quad \Delta_N^\perp:=\Pi_{<
N}\Delta,\quad  Q_N^\perp:=\Pi_{< N}B.
$$

Let $(\zeta(t))_{t\ge0}$ be the solution of the problem
\begin{equation}\label{ET211}
\begin{aligned}
\frac{\d\zeta}{\d t}(t)&=-\Delta_N\zeta(t)+ \Pi_{\ge N}B_s(\om(t;w),
\zeta(t))  -\frac12\zeta_{N}(t)|\zeta_{N}(t)|^{-1},\\
&\\
\zeta(0)&=\xi,
 \end{aligned}
\end{equation}
for a given  integer $N\ge 1$. Here $\zeta_{N}(t):=\Pi_{< N}\zeta(t)$. We adopt the convention that
\begin{equation}\label{ET212}
\zeta_{N}(t)|\zeta_{N}(t)|^{-1}:=0\qquad
\text{if\quad  $\zeta_{N}(t)=0$.}
\end{equation}
Let
\begin{equation}\label{083005}
g:=Q^{-1}f,
\end{equation}
where
\begin{equation}\label{ET213}
f(t):=-\Delta_N^{\perp}\zeta(t)+ \Pi_{<N}B_s(\om(t),\zeta(t)) +
\frac12\zeta_N(t)|\zeta_N(t)|^{-1}.
\end{equation}
Note that $f$ takes values in a finite dimensional space. Recall that $\rho(t)=\xi(t)- D(t)$.
The proof of the proposition in question shall be achieved at the end of several auxiliary facts formulated as lemmas.
\begin{lemma} \label{L2}
We have
\begin{equation}
\rho(t) =
\zeta(t),\qquad\forall\,t\ge0.
\end{equation}
\end{lemma}
\begin{proof}
 Adding $f(t)$ to the both sides of \REQ{ET211} we obtain
\begin{equation}\label{ET215}
\begin{aligned}
 \frac{\d \zeta(t)}{\d t}(t) +f(t)&=-\Delta\zeta(t)+ B_s(\om(t),\zeta(t)),\qquad
 \zeta(0)&=\xi.
 \end{aligned}
\end{equation}
Recall that $\xi(t)$ and $D(t)$ satisfy equations
$\REQ{E25-1}$ and \REQ{ET28}, respectively. Hence $\rho(t)$
satisfies
$$
\begin{aligned}
\frac{\d \rho(t)}{\d t} &= -\Delta\rho_t + B_s(\om(t),\rho(t))  - Qg(t),\\
\rho(0) &= \xi.
\end{aligned}
$$
Since, $f(t)=Qg(t)$ we conclude that $\rho(t)$ and $\zeta(t)$ solve the same linear evolution
 equation with the same initial value. Thus the assertion of the lemma follows.
\end{proof}

\begin{lemma}\label{L3}
For each $N\ge1$ we have 
\begin{equation}
\label{011312}
\zeta_N(t)=0,\quad \forall\,t\ge 2,
\end{equation}
 and 
 \begin{equation}
 \label{021312}
 |\zeta_N(t)|\le 1,\quad \forall \,t\ge0.
 \end{equation}
\end{lemma}
\begin{proof} By Lemma \ref{L2} we have $\rho(\cdot)=\zeta(\cdot)$. Applying $\Pi_{<N}$ to both
sides of \REQ{ET211} we obtain
\begin{equation}\label{ET216}
\begin{aligned}
 \frac{\d~}{\d t}\zeta_N(t)&=-\frac12 |\zeta_N(t)|^{-1}\zeta_N(t),\\
 \zeta(0)&=\xi.
 \end{aligned}
\end{equation}
Multiplying both sides of \REQ{ET216} by $\zeta_N(t)$ we
obtain that $z(t):=|\zeta_N(t)|^2$ satisfies
\begin{equation}\label{ode}
\frac{\d z}{\d t}(t)=-\frac12 \sqrt{z(t)},\quad z(0)=|\xi|^2.
\end{equation}
Since $0\le z(0)\le1$ the desired conclusion
holds from elementary properties of the solution of o.d.e. \eqref{ode}.
\end{proof}

Let $\zeta^{(N)}(t):=\Pi_{\ge N}\zeta(t)$. We have
\begin{equation}\label{ET211N}
\frac{\d}{\d t}|\zeta^{(N)}(t)|^2=-2\|\zeta^{(N)}(t)\|^2+ \langle\Pi_{\ge N}B_s(\om(t;w),
\zeta(t)),\zeta^{(N)}(t)\rangle,\quad
\zeta^{(N)}(0)=\Pi_{\ge N}\xi.
\end{equation}
We shall use the following estimates, see Proposition 6.1 of \cite{foias}. There exists $C>0$ such that
\begin{equation}
\label{cofo88}
|\langle B_0(h,\om_1),\om_2)\rangle|\le C|h|_{s_1-1}|\om_1|_{1+s_2}|\om_2|_{s_3} ,\quad\forall\, h\in H_{s_1-1},\om_1\in H_{1+s_2},\om_2\in H_{s_3},
\end{equation}
for all $s_1,s_2,s_3\ge0$ such that $s_1+s_2+s_3>1$. When, in addition 
 $s_1>1$ we have
\begin{equation}
\label{cofo88a}
|\langle B_1(h,\om_1),\om_2)\rangle|\le C|h|_{s_1-1}|\om_1|_{1+s_2}|\om_2|_{s_3} ,\quad\forall\, h\in H_{s_1-1},\om_1\in H_{1+s_2},\om_2\in H_{s_3}.
\end{equation}
With the help of the above inequalities   we can bound the symmetric part of the bilinear form $B(\cdot,\cdot)$ as follows.
\begin{lemma}
\label{lm011412}
For any $\varepsilon\in(0,1)$ and $\nu>0$  there exists a constant $C>0$ such that 
\begin{eqnarray}
\label{021412}
&&|\langle B_s(\om(t;w),
\zeta(t)),\zeta^{(N)}(t)\rangle|\\
&&
\le (\varepsilon N+C+\frac{\nu}{2}\|\om(t;w)\|^2)|\zeta^{(N)}(t)|^2+\frac{1}{4}\|\zeta^{(N)}(t)\|^2+C\|\om(t;w)\|^2|\zeta_{N}(t)|^2.\nonumber
\end{eqnarray}
\end{lemma}
\proof From \eqref{cofo88} we have
\begin{eqnarray*}
&&|\langle B_0(\om(t;w),
\zeta(t)),\zeta^{(N)}(t)\rangle|=|\langle B_0(\om(t;w),
\zeta^{(N)}(t)), \zeta(t)\rangle|\\
&&
\le C|\om(t;w)|_{1/2}\|\zeta^{(N)}(t)\||\zeta(t)|\le \frac{1}{16} \|\zeta^{(N)}(t)\|^2+C_1|\om(t;w)|_{1/2}^2|\zeta(t)|^2.
\end{eqnarray*}
Using the  Gagliardo--Nirenberg  and Young's inequalities we get
$$
C_1|\om(t;w)|_{1/2}^2\le \frac{\nu}{8} \|\om(t;w)\|^2 +C_2|\om(t;w)|^2
$$
for some $C_2>0$. This yields
$$
|\langle B_0(\om(t;w),
\zeta(t)),\zeta^{(N)}(t)\rangle|
\le  \frac{1}{16} \|\zeta^{(N)}(t)\|^2+ \frac{\nu}{8} \|\om(t;w)\|^2 +C_2|\om(t;w)|^2.
$$
Likewise,
 \begin{eqnarray*}
|\langle B_0(
\zeta(t),\om(t;w)),\zeta^{(N)}(t)\rangle|&&\le C\|\om(t;w)\||\zeta^{(N)}(t)|_{1/2}|\zeta(t)|\\
&&
\le \frac{\nu}{8} \|\om(t;w)\|^2|\zeta(t)|^2+C_1|\zeta^{(N)}(t)|_{1/2}^2\\
&&
\le \frac{\nu}{8} \|\om(t;w)\|^2|\zeta(t)|^2+ \frac{1}{16}\|\zeta^{(N)}(t)\|^2+C_2|\zeta^{(N)}(t)|^2.
\end{eqnarray*}
On the other hand
$$
|\langle B_1(\om(t;w),
\zeta(t)),\zeta^{(N)}(t)\rangle|=|\langle B_1(\om(t;w),
\zeta^{(N)}(t)), \zeta^{(N)}(t)\rangle|=0
$$
and
\begin{equation}
\label{051412}
|\langle B_1(
\zeta(t),\om(t;w)),\zeta^{(N)}(t)\rangle|\le C\|\om(t;w)\||\zeta^{(N)}(t)||\zeta(t)|_{1/2}.
\end{equation}
Note that 
$$
|\zeta(t)|_{1/2}^2=|\zeta_N(t)|_{1/2}^2+|\zeta^{(N)}(t)|_{1/2}^2\le N|\zeta_N(t)|^2+|\zeta^{(N)}(t)|_{1/2}^2. 
$$
With this inequality we can estimate of the right hand side of \eqref{051412} by
$$
CN^{1/2}\|\om(t;w)\||\zeta^{(N)}(t)||\zeta_N(t)|+C\|\om(t;w)\||\zeta^{(N)}(t)||\zeta^{(N)}(t)|_{1/2}.
$$
The first term can be estimated by
$$
C_1\|\om(t;w)\|^2|\zeta_N(t)|^2+\varepsilon N|\zeta^{(N)}(t)|^2.
$$
The  second term is less than, or equal to
\begin{eqnarray*}
&&\frac{\nu}{8}\|\om(t;w)\|^2|\zeta^{(N)}(t)|^2+C_1|\zeta^{(N)}(t)|_{1/2}^2\\
&&\quad 
\le \frac{\nu}{8}\|\om(t;w)\|^2|\zeta^{(N)}(t)|^2+ \frac{1}{16}\|\zeta^{(N)}(t)\|^2 +C_2|\zeta^{(N)}(t)|^2 .
\end{eqnarray*}
Summarizing the above consideration we have shown \eqref{021412}.
\qed

\subsubsection*{Proof of \eqref{ET210}}

Performing the scalar product in $H$ of both sides of \eqref{ET211} against $\zeta^{(N)}(t)$ and using Lemma \ref{L2} we conclude that
\begin{equation}\label{ET211b}
\begin{aligned}
\frac{\d}{\d t}|\zeta^{(N)}(t)|^2&\le-2\|\zeta^{(N)}(t)\|^2+2 (\epsilon N+C+\frac{\nu}{2}\|\om(t;w)\|^2)|\zeta^{(N)}(t)|^2\\
&\quad 
+\frac12\|\zeta^{(N)}(t)\|^2+C\|\om(t;w)\|^2|\zeta_{N}(t)|^2\\
&\le -\frac12\|\zeta^{(N)}(t)\|^2+\left[-N^2+2 (\epsilon N+C)\right]|\zeta^{(N)}(t)|^2\\
&\quad 
+\nu\|\om(t;w)\|^2|\zeta^{(N)}(t)|^2+C\|\om(t;w)\|^2|\zeta_{N}(t)|^2, \\
\zeta(0)&=\xi.
 \end{aligned}
\end{equation}
Suppose that $N_0$ 
is such that 
\begin{equation}
\label{N0}
N^2_0-2 (\varepsilon N_0+C)\ge \max\{N^2_0/2,\gamma+{\rm tr}\,Q^2\}.
\end{equation}
 Then, solve \eqref{ET211} and determine $g(t)$ via \eqref{083005}. According to Lemma \ref{L2} the difference $\rho(t)=\xi(t)-D(t)$
equals $\zeta(t)$. From \eqref{ET211b} we conclude via Gronwall's inequality that
\begin{equation}\label{ET211c}
\begin{aligned}
&|\zeta^{(N_0)}(t)|^2\le|\zeta^{(N_0)}(0)|^2\exp\left\{-\gamma t-{\rm tr}\,Q^2t+\nu\int_0^t\|\om(s;w)\|^2\d s\right\}\\
&
+C\int_0^t\exp\left\{-\gamma(t-s)-{\rm tr}\,Q^2(t-s)+\nu\int_s^t\|\om(r;w)\|^2\d r\right\}\|\om(s;w)\|^2|\zeta_{N_0}(s)|^2\d s, \\
&\zeta(0)=\xi.
 \end{aligned}
\end{equation}
From Lemma \ref{L3} the second term on the right hand side of  \eqref{ET211c} can be estimated by
\begin{eqnarray*}
&&
C\exp\left\{-\gamma(t-2)-{\rm tr}\,Q^2(t-2)\right\}\exp\left\{\nu\int_0^t\|\om(r;w)\|^2\d r\right\}\int_0^2\|\om(s;w)\|^2\d s\\
&&\quad 
\le C_1\exp\left\{-\gamma(t-2)-{\rm tr}\,Q^2(t-2)\right\}\exp\left\{\nu'\int_0^{t\vee 2}\|\om(r;w)\|^2\d r\right\},
\end{eqnarray*}
provided that $\nu'>\nu$. Therefore
\begin{equation}\label{ET211d1}
\begin{aligned}
&|\zeta^{(N_0)}(t)|^2\le|\zeta^{(N_0)}(0)|^2\exp\left\{-\gamma t-{\rm tr}\,Q^2t+\nu\int_0^{t\vee 2}\|\om(s;w)\|^2\d s\right\}\\
&\quad 
+ C_1\exp\left\{-\gamma(t-2)-{\rm tr}\,Q^2(t-2)\right\}\exp\left\{\nu'\int_0^{t\vee 2}\|\om(r;w)\|^2\d r\right\},\quad\bbP\,{\rm a.s.} \\
&\zeta(0)=\xi
 \end{aligned}
\end{equation}
for all $t>0$.
Estimate \eqref{ET210}, with $e_{\nu'}(w)$ appearing on the right hand side, is then a
consequence of the above bound, Lemma \ref{L2} and estimate \eqref{010712-00} if only $0<\nu<\nu'<\nu_0$.

\subsubsection*{Proof of \eqref{ET29}}
To prove the estimate observe that from \eqref{083005}, \eqref{ET213} and \eqref{011312} it follows that
$$
|g(t)|= |Q^{-1}\Pi_{<{N_0}}B_s(\om(t),\zeta(t))|\le |g_0(t)|+|g_1(t)|,\quad\forall\,t\ge 0,
$$
with
$$
g_i(t):= Q^{-1}\Pi_{<{N_0}}B_{i,s}(\om(t),\zeta(t)),\quad i=0,1.
$$

\subsubsection{Estimates of $|g_1(t)|$}
Note that for $t\ge2$,
$$
|g_1(t)|= |Q^{-1}\Pi_{<{N_0}}B_1(\zeta^{(N_0)}(t),\om(t))|\le C\|\zeta^{(N_0)}(t)\|\|\Pi_{<{N_0}}\om(t)\le CN_0\|\zeta^{(N_0)}(t)\||\om(t)|.
$$
The last inequality holds because 
\begin{equation}
\label{012402}
\|\Pi_{<{N_0}}w\|\le N_0|w|,\quad\forall\, w\in H.
\end{equation}
Therefore
$$
\bbE\int_2^T|g_1(t)|^2d t\le C J(T),
$$
with
$$
J(T):=\bbE\int_2^T\|\zeta^{(N_0)}(t)\|^2|\om(t)|^2d t.
$$
We use \eqref{ET211b} to get 
\begin{equation}\label{ET211d}
\begin{aligned}
&J(T)\le -2\bbE\int_2^T\frac{d}{d t}|\zeta^{(N_0)}(t)|^2|\om(t)|^2d t
+2\nu\bbE\int_2^T|\om(t)|^2\|\om(t;w)\|^2|\zeta^{(N_0)}(t)|^2d t.
 \end{aligned}
\end{equation}
Denote the terms appearing on the right hand side as $J_i(T)$, $i=1,2$, respectively.
We have
\begin{eqnarray*}
&&J_1(T)=2\bbE\int_2^T|\zeta^{(N_0)}(t)|^2d|\om(t)|^2-2\bbE|\zeta^{(N_0)}(t)|^2|\om(t)|^2\Big|_2^T.
\end{eqnarray*}
The boundary term appearing on the right hand side is easily estimated by $Ce_{\nu}(w)$, by virtue of \eqref{ET211d1} and \eqref{010712}. As for the integral term,
using \eqref{051512} and the already proven \eqref{ET210}, we can estimate it by
\begin{eqnarray*}
&&2\bbE\int_2^T|\zeta^{(N_0)}(t)|^2\left({\rm tr}\, Q^2-2\|\om(t)\|^2\right)d t\le Ce_{\nu}(w)\int_2^T\e^{-\gamma t}d t\le C_1e_{\nu}(w).
\end{eqnarray*}
Next, we can write
\begin{eqnarray*}
&&J_2(T)\le J_{21}(T)+J_{22}(T),
\end{eqnarray*}
where
\begin{eqnarray*}
&&J_{21}(T):=2\nu|\zeta^{(N_0)}(0)|^2\bbE\int_2^T|\om(t)|^2\|\om(t;w)\|^2\e^{-(\gamma +{\rm tr}\,Q^2)t}\exp\left\{\nu\int_0^{t}\|\om(s;w)\|^2d s\right\}d t\\
&&
J_{22}(T):=2\nu C_1\bbE\int_2^T|\om(t)|^2\|\om(t;w)\|^2\e^{-(\gamma +{\rm tr}\,Q^2)t}\exp\left\{\nu'\int_0^{t}\|\om(r;w)\|^2d r\right\}d t.
\end{eqnarray*}
Observe that
\begin{eqnarray*}
&&J_{21}(T)\le 2\bbE\int_2^T|\om(t)|^2\e^{-(\gamma +{\rm tr}\,Q^2)t}\frac{d}{d t}\exp\left\{\nu\int_0^{t}\|\om(s;w)\|^2d s\right\}d t=\sum_{i=1}^3J_{21i}(T),
\end{eqnarray*}
where
\begin{eqnarray*}
&&J_{211}(T) :=2\e^{-(\gamma +{\rm tr}\,Q^2)t}\bbE\left\{|\om(t)|^2\exp\left\{\nu\int_0^{t}\|\om(s;w)\|^2d s\right\}\right\}\Big|_2^T,\\
&&J_{212}(T) :=2(\gamma +{\rm tr}\,Q^2)\int_2^T\e^{-(\gamma +{\rm tr}\,Q^2)t}\bbE\left\{|\om(t)|^2\exp\left\{\nu\int_0^{t}\|\om(s;w)\|^2d s\right\}\right\}dt,\\
&&J_{213}(T):=-2\int_2^T \e^{-(\gamma +{\rm tr}\,Q^2)t}\bbE\left\{\exp\left\{\nu\int_0^{t}\|\om(s;w)\|^2d s\right\}d |\om(t)|^2\right\}.
\end{eqnarray*}
We have
$$
J_{211}(T)\le C\e^{-(\gamma +{\rm tr}\,Q^2)T}\bbE \exp\left\{\nu|\om(T;w)|^2+\nu\int_0^{T}\|\om(s;w)\|^2d s\right\}\le C_1\e^{-\gamma T}e_{\nu}(w).
$$
The last inequality follows from \eqref{010712-00}. On the other hand, by the same token
\begin{eqnarray*}
J_{212}(T)&&\le C\bbE\int_2^T\e^{-(\gamma +{\rm tr}\,Q^2)t}\exp\left\{\nu|\om(t)|^2+\nu\int_0^{t}\|\om(s;w)\|^2d s\right\}d t\\
&&
\le C_2e_{\nu}(w)\int_2^T\e^{-\gamma t}d t\le C_3e_{\nu}(w),\quad\forall\,T\ge 2
\end{eqnarray*}
and finally
\begin{eqnarray*}
J_{213}(T)&&\le C\bbE\int_2^T\|\om(t)\|^2 \e^{-(\gamma +{\rm tr}\,Q^2)t}\exp\left\{\nu\int_0^{t}\|\om(s;w)\|^2d s\right\}d t\\
&&
\le \frac{C}{\nu}\bbE\int_2^T \e^{-(\gamma +{\rm tr}\,Q^2)t}\frac{d }{d t}\exp\left\{\nu\int_0^{t}\|\om(s;w)\|^2d s\right\}d t.
\end{eqnarray*}
Repeating the integration by parts  argument used before we conclude that also
$$
J_{213}(T)\le Ce_{\nu}(w),\quad\forall\,T\ge2.
$$
Summarizing, we have shown that
$
J_{21}(T)\le Ce_{\nu}(w),
$ for $T\ge2$.
In the same way we can argue that
$
J_{22}(T)\le Ce_{\nu}(w),
$
thus also
$$
J_{2}(T)\le Ce_{\nu}(w),\quad\forall\,T\ge2.
$$

Finally, for $t\in[0,2]$ we  use \eqref{012402} to obtain that
\begin{eqnarray*}
|g_1(t)|&&= |Q^{-1}\Pi_{<{N_0}}B_{1,s}(\zeta(t),\om(t))|\\
&&
\le C\left(\|\zeta^{(N_0)}(t)\|+N_0|\zeta_{N_0}(t)|\right)|\om(t)|+CN_0\|\om(t)\||\zeta_{N_0}(t)|.
\end{eqnarray*}
We have therefore
$$
\int_0^2\bbE |g_1(t)|^2d t\le J_{31}+J_{32},
$$
with
\begin{eqnarray*}
&&
J_{31}:=C\int_0^2\bbE\|\zeta^{(N_0)}(t)\|^2|\om(t)|^2d t,\\
&&
J_{32}:=C\int_0^2\bbE(|\om(t)|^2+\|\om(t)\|^2)d t.
\end{eqnarray*}
It is easy to see from \eqref{051512} that
$
J_{32}\le Ce_{\nu}(w).
$
Term $J_{31}$ satisfies an estimate analogous to \eqref{ET211d},
we can write therefore that
$$
J_{31}\le J_{311}+J_{312},
$$
where $J_{311}$, $J_{312}$ are defined as the corresponding expression on the right hand side of \eqref{ET211d}  
 with the limits of the integrals appearing on the right hand side  replaced by $0$ and $2$ correspondingly. In the case of $J_{311}$ we  proceed in the same way for $J_1(T)$ and end up with the bound
$
J_{311}\le Ce_{\nu}(w).
$
On the other hand, from \eqref{ET211c} we get
\begin{eqnarray}
\label{011612}
&&
J_{312}\le C\bbE\int_0^2\|\om(t;w)\|^2|\om(t)|^2\exp\left\{\nu\int_0^t\|\om(s)\|^2d s\right\}d t\\
&&\qquad 
+C\bbE\int_0^2\int_0^t\|\om(s)\|^2|\om(t)|^2\|\om(t)\|^2\exp\left\{\nu\int_s^t\|\om(r)\|^2d r\right\}d td s.\nonumber
\end{eqnarray}
Repeating the argument with the integration by parts we have used in the foregoing we conclude that the first term on the right hand side is estimated by
$e_{\nu}(w).$ The second term equals
\begin{eqnarray*}
&&
-\frac{C}{\nu}\bbE\int_0^2|\om(t)|^2\|\om(t)\|^2d t\int_0^t\frac{d }{d s}\exp\left\{\nu\int_s^t\|\om(r)\|^2d r\right\}d s\\
&&\qquad 
\le \frac{C}{\nu}\bbE\int_0^2|\om(t)|^2\|\om(t)\|^2\exp\left\{\nu\int_0^t\|\om(r)\|^2d r\right\} d t.
\end{eqnarray*}
From here on we estimate as in the foregoing and conclude that this term is less than $e_{\nu}(w).$
Summarizing, we have shown  that
$$
J(T)\le Ce_{\nu}(w),\quad\forall\,T\ge2.
$$

\subsubsection{Estimates of $g_0(t)$}
We start with the following.
\begin{lemma}
\label{form1} (cf. Lemma A.1 of  \cite{EMS01})
For any $N$ there exists $C_N$ such that
$$
 |\Pi_{<N}B_{0}(h,\om)|\le C_N|h|_{-1}|\om|,\quad\forall\,h\in H_{-1},\,\om\in H.
$$
\end{lemma}
\proof
Suppose that
$$
h=\sum_{k\in\bbZ^2_*}\hat h(k)e_k,\quad \om=\sum_{k\in\bbZ^2_*}\hat \om(k)e_k.
$$
We can write that
\begin{eqnarray*}
|\Pi_{<N}B_0(h,\om)|^2&&=\int_{\bbT}|\Pi_{<N}\nabla\cdot({\cal K}(h)(x)\om(x))|^2d x\\
&&
\le N^2\sum_{0<|k|<N}\left|\sum_{\ell\in\bbZ^2_*}\widehat{{\cal K}(h)}(\ell)\hat\om(k-\ell)\right|^2\le N^4|h|^2_{-1}|\om|^2.\qed
\end{eqnarray*}

From the above lemma we get that for $T\ge2$,
$$
\bbE\int_2^T|g_0(t)|^2d t\le C I(T)
$$
with
\begin{eqnarray*}
I(T)&&:=\bbE\int_2^T|\zeta^{(N_0)}(t)|^2|\om(t)|^2d t\\
&&
\le C\int_2^T \bbE\exp\left\{-\gamma t-{\rm tr}\,Q^2t+\nu|\om(t)|^2+\nu\int_0^t\|\om(s)\|^2d s\right\}d t\\
&&
+C \bbE\int_2^T\exp\left\{-\gamma (t-2)-{\rm tr}\,Q^2(t-2)+\nu'|\om(t)|^2+\nu'\int_0^t\|\om(s)\|^2d s\right\}d t
\\
&&
\le C_1e_{\nu'}(w),
\end{eqnarray*}
provided $0<\nu<\nu'<\nu_0$. 
The first inequality follows from \eqref{ET211d1}, while the second from \eqref{010712-00}.
This,  ends the proof of Proposition \ref{m-lm} and according to our previous remarks concludes  the verification
of Assumption 2.

\subsection{Assumption 3} To verified  this assumption consider the solution $y(t;w)$, $t\ge 0$,   to  the deterministic equation
$$
\frac{d y (t)}{d t}=\Delta y(t)+B(y(t)), \qquad t \ge 0, 
$$
with the initial condition $y(0)=w$. Then 
$$
\lim_{t\to+\infty}\sup_{|w|\le R}|y(t;w)|=0,\quad\forall\,R>0.
$$
Fix $\delta>0$ and $R>0$. Let $T_0>0$ be such that 
$$
\sup_{|w|\le R}{\rm d}_{\kappa}(y(T; w),0)\le\delta/4,\quad\forall\,T\ge T_0.
$$ 
 Since 
 $$W_{\Delta, Q}(t):=   \int_0^t \e ^{\Delta (t-s)} Qd W(s)
 $$ is a centered Gaussian random element in the Banach space $C([0, T]; V)$ with the uniform norm 
 $$
 \|f\|_{\infty,T}:=\sup_{t\in[0,T]}\|f(t)\|,\quad f\in C([0, T]; V),
 $$   its topological support is a closed linear subspace (see e.g.  \cite{VAK}). Thus, in particular, $0$ belongs to the support of its law and for any $\varrho>0$, $\mathbb{P}(F_\varrho)>0$, where 
$$
F_{\varrho}=\{\pi\in\Omega\colon  \|W_{A, Q}(\pi)\|_{\infty,T}<\varrho\}. 
$$
Choose $\varrho_0>0$ such that 
$$
{\rm d}_{\kappa}(\om(T; w_i)(\pi),y(T; w_i))|\le\delta/4\qquad\text{for all $\pi\in F_{\varrho_0}$, $i=1,2$ and $|w|\le R$,}
$$
and set $a:=\mathbb{P}(F_{\varrho_0})>0$. For any  $|w_1|,|w_2|\le R$ we have 
$$
{\cal P}_{T}(w_1,w_2;\Delta_{\delta,\kappa})\ge\mathbb{P}\left[\pi\in\Omega\colon  {\rm d}_{\kappa}(\om(T; w_i)(\pi),y(T; w_i))|\le\delta/4,\,i=1,2\right]\ge \mathbb{P}(F_{\varrho_0})=a, 
$$
and thus we have finished verification of Assumption 3. 
$\qed$

\subsection{Proof of Theorem \ref{T2}}

\subsubsection{Proof of part 1)}  

Let us fix an arbitrary $T>0$ and define
$
\zeta(t):=|\om(t)|^2+t\|\om(t)\|^2
$ and ${\rm tr}\,Q_1:=\sum_{k\in\bbZ^2_*}|k|^2q_k^2$.
By It\^o's formula we have
\begin{equation}
\label{012502}
d\zeta(t)=\left[{\rm tr}\, Q^2+t{\rm tr}\,Q_1-2t|\om(t)|^2_2-\|\om(t)\|^2+2t\langle B(\om(t)),\Delta\om(t)\rangle\right]d t
+dM_t
\end{equation}
and
$$
dM_t:=2\langle Qd W(t), (I+t\Delta)\om(t)\rangle.
$$
According to  \eqref{020712} there exist $C,C_1>0$ such that
$$
|\langle B_0(\om),\Delta\om\rangle|\le C|\om|_{1/2}\|\om\||\om|_2\le\frac14 |\om|_2^2+C_1|\om|^4,\quad\forall\, \om\in H_2.
$$
Likewise, from \eqref{cofo88a} with $s_1=3/2$, $s_2=s_3=0$, we have
$$
|\langle B_1(\om),\Delta\om\rangle|\le C|\om|_{1/2}\|\om\||\om|_2,\quad\forall\, \om\in H_2.
$$
With these inequalities we conclude that
$$
|\langle B(\om),\Delta\om\rangle| \le\frac12 |\om|_2^2+C_1|\om|^4,\quad\forall\, \om\in H_2.
$$
From here on we proceed as in the proof of Lemma A.3 of \cite{MP} and conclude from \eqref{012502} that
\begin{equation}
\label{032502}
\zeta(t)\le |w|^2+t{\rm tr}\, Q^2+\frac{t^2{\rm tr}\,Q_1}{2}+C\int_0^ts|\om(s)|^4ds+U(t),
\end{equation}
where $U(0)=0$ and
$$
dU(t)=-(t|\om(t)|^2_2+\|\om(t)\|^2)dt+dM_t.
$$
Since 
$$
U(t)\le M_t-(\alpha/2) \langle M\rangle_t
$$
 for some sufficiently small $\alpha>0$ we conclude from the exponential martingale inequality that
$$
\bbP[\sup_{t\in[0,T]}U(t)\ge K]\le \e^{-\alpha K},\quad\forall\,K>0.
$$
This, of course, implies that $\bbE \exp\left\{\alpha' \sup_{t\in[0,T]}U(t)\right\}<+\infty$ for any $\alpha'\in(0,\alpha)$.
From \eqref{010712-00} we get
$$
\bbE\exp\left\{\nu\sup_{t\in[0,T]}|\om(t)|^2\right\}\le Ce_{\nu}(w),
$$
which in turn implies that  
$$
\bbE\left[  \sup_{t\in[0,T]}|\om(t)|^{4N}\right]\le C|w|^{4N}.
$$ Summarizing, 
the above consideration we obtain from \eqref{032502} that
for any $T>0$ and $N\ge0$ there exists a constant $C>0$ such that
\begin{equation}
\label{011912}
\bbE\left[\sup_{s\in[0,T]}\zeta^{2N}(s)\right]\le C\left(|w|^{4N}+1\right).
\end{equation}
Thus we conclude the proof of part 1) of Theorem \ref{T2}.

\subsubsection{Proof of part 2)}
First note that $P_t\phi(w)$ is well defined thanks to the already proved estimate \eqref{E27-aa} and the definition of the norm $\|\!|\cdot\|\!|_N$. In addition, we have
\begin{equation}
\label{031912}
e_{-\nu}(w)|P_t\phi(w)|\le \|\!|\phi\|\!|_Ne_{-\nu}(w)(1+\bbE\|\om(t;w)\|^N)\le C\|\!|\phi\|\!|_N,\quad \forall\,w\in H.
\end{equation}
To deal with $DP_t\phi(w)[\xi]$ we first show the following:
\begin{lemma}
\label{lm031912}
Suppose that $\{\xi(t),\,t\ge0\}$ is defined by \eqref{041912}. Then, for any $t,\nu>0$ there exists $C>0$ such that
\begin{equation}
\label{021912}
\|\xi(t)\|^2\le C\|\xi\|^2\exp\left\{\nu\int_0^t\|\om(s;w)\|^2d s+Ct\right\},\quad\forall\,t\ge 0,\,w\in H,\,\xi\in V,\,\bbP-{\rm a.s.}
\end{equation}
\end{lemma}
\proof
Let $\zeta(t):=|\xi(t)|^2+\gamma\|\xi(t)\|^2$, with $\gamma>0$ to be chosen later on. We have
$$
\partial_t\zeta(t)=-2\|\xi(t)\|^2-2\gamma|\xi(t)|^2_2+\gamma\langle B_{s}(\xi(t),\om(t)),\Delta \xi(t)\rangle
+\langle B(\xi(t),\om(t)),\xi(t)\rangle.
$$
Thanks to \eqref{cofo88} with $s_1=3/2$, $s_2=s_3=0$ we can   find  constants $C,C_1>0$ such that
\begin{eqnarray*}
\gamma|\langle B_{0}(\xi(t),\om(t)),\Delta \xi(t)\rangle|&&\le C\gamma|\xi(t)|_2|\xi(t)|_{1/2}\|\om(t)\|
\\
&&
\le \frac{1}{4}|\xi(t)|_2^2+C_1\gamma^2|\xi(t)|\|\xi(t)\|\|\om(t)\|^2\\
&&
\le \frac{1}{4}|\xi(t)|_2^2+\nu|\xi(t)|^2\|\om(t)\|^2+\frac{C_2\gamma^4}{\nu}
\|\xi(t)\|^2\|\om(t)\|^2.
\end{eqnarray*}
Using again \eqref{cofo88}, this time with $s_1=2$, $s_2=s_3=0$, we obtain
\begin{eqnarray*}
\gamma|\langle B_{0}(\om(t),\xi(t)),\Delta \xi(t)\rangle|&&\le C\gamma|\xi(t)|_2\|\xi(t)\|\|\om(t)\|
\\
&&
\le \frac{1}{4}|\xi(t)|_2^2+C_1\gamma^2\|\xi(t)\|^2\|\om(t)\|^2.
\end{eqnarray*}
Also  from \eqref{cofo88a}, used with $s_1=3/2$, $s_2=s_3=0$, we obtain
\begin{eqnarray*}
\gamma|\langle B_1(\xi(t),\om(t)),\Delta \xi(t)\rangle|&&\le C\gamma|\xi(t)|_2|\xi(t)|_{1/2}\|\om(t)\|
\\
&&
\le \frac{1}{4}|\xi(t)|_2^2+\nu|\xi(t)|^2\|\om(t)\|^2+\frac{C_2\gamma^4}{\nu}
\|\xi(t)\|^2\|\om(t)\|^2.
\end{eqnarray*}
In addition,
$$
\gamma|\langle B_1(\om(t),\xi(t)),\Delta \xi(t)\rangle|=0.
$$
On the other hand,
\begin{eqnarray*}
|\langle B_{0}(\xi(t),\om(t)),\xi(t)\rangle|&&\le C|\xi(t)||\xi(t)|_{1/2}\|\om(t)\|
\\
&&
\le \nu |\xi(t) |^2|\|\om(t)\|^2+C_1|\xi(t)|_{1/2}^2\\
&&
\le \nu |\xi(t) |^2|\|\om(t)\|^2+\frac{1}{4}\|\xi(t)\|^2+C_2|\xi(t)|^2,
\end{eqnarray*}
and
\begin{eqnarray*}
|\langle B_{1}(\xi(t),\om(t)),\xi(t)\rangle|&&\le C|\xi(t)||\xi(t)|_{1/2}\|\om(t)\|
\\
&&
\le \nu |\xi(t) |^2|\|\om(t)\|^2+C_1|\xi(t)|_{1/2}^2\\
&&
\le \nu |\xi(t) |^2|\|\om(t)\|^2+\frac{1}{4}\|\xi(t)\|^2+C_2|\xi(t)|^2.
\end{eqnarray*}
Summarizing,  for a sufficiently small $\gamma>0$ and some constant $C>0$ we can write that
$$
\partial_t\zeta(t)\le \left(\nu\|\om(t)\|^2+C\right)\zeta(t)
$$
and \eqref{021912} follows by Gronwall's inequality.
\qed

Concerning the estimates of $|DP_t\phi(w)[\xi]|$ we can write that 
\begin{eqnarray}
\label{051912}
&&
e_{-\nu}(w)|DP_t\phi(w)[\xi]|=e_{-\nu}(w)\left|\bbE\left[(D\phi)(\om(t;w))[\xi(t)]\right]\right|\\
&&
\le \|\!|\phi\|\!|_Ne_{-\nu}(w)\bbE\left[(1+\|\om(t;w)\|^N)\|\xi(t)\|\right]\nonumber\\
&&
\le C\|\!|\phi\|\!|_Ne_{-\nu}(w)\left\{\bbE(1+\|\om(t;w)\|)^{2N}\right\}^{1/2}\left\{\bbE\|\xi(t)\|^2\right\}^{1/2},\quad \forall\,w\in H.
\end{eqnarray}
By the already proved part 1) of the theorem and Lemma  \ref{lm031912} we obtain that the utmost right hand side is less than, or equal to
$$
  C_1\|\xi\|\|\!|\phi\|\!|_Ne_{-\nu}(w)(1+|w|^{4N})\bbE\exp\left\{\frac{\nu}{2}\int_0^t\|\om(s;w)\|^2d s+C_1t\right\}\le C_2\|\xi\|\|\!|\phi\|\!|_N.
  $$
Hence
$$
e_{-\nu}(w)\|DP_t\phi(w)\|\le  C_2\|\!|\phi\|\!|_N
$$
and thus we have finished the proof of part 2) of Theorem \ref{T2}.

\appendix

\section{Existence of the Markov, Feller family}

\label{secAp1}

\subsubsection*{Proof of Theorem \ref{MF1}} Given $N\in \mathbb{N}$, denote by $\Pi_N$ the orthogonal projection of $H$ into $H_N:=\text{\rm span}\,\{e_k,\,0<|k|\le N\}$. 
Consider the following  finite dimensional It\^o stochastic differential equation 
 \begin{equation}\label{E25N}
 \begin{aligned}
d \om^{(N)}(t) &=[\Delta\om^{(N)}(t) -B_0^{(N)}(\om^{(N)}(t))-B_1^{(N)}(\om^{(N)}(t))]d t + Q^{(N)}d  W(t), \\
 \om^{(N)}(0)&=w^{(N)}\in H, 
\end{aligned}
\end{equation}
with $W^{(N)}(t):=\Pi_N W(t)$, $Q^{(N)}:=\Pi_N Q$,  and 
$$
B_0^{(N)}(\om):=\Pi_NB_0(\om),\quad B_1^{(N)}(\om):=\Pi_NB_1(\om),\quad\om\in H_N.
$$
The local existence and uniqueness of solution to \eqref{E25N} follows from a result for finite dimensional S.D.E.-s. By It\^o's formula we get the  estimate
\begin{equation}\label{E26N}
 \mathbb{E}\left\{ |\om^{(N)}(T)|^2 +\frac{1}{2}\int_0^T \|\om^{(N)}(t)\|^2 d t \right\} \le  |w^{(N)}|^2 +  \|Q^{(N)}\|^2_{L_{(HS)(H,H)}} T
\end{equation} 
From this we conclude that the sequence $\{\om^{(N)}(t),\,t\in[0,T]\}$, $N\ge1$ is compact in 
$
 L^2(\Omega, {\cal E}_T).
$
In addition,
\begin{align*}
\om^{(N)}(t) &= \e ^{\Delta t} w^{(N)} - \int_0^t \e ^{\Delta (t-s)} B_0^{(N)}(\om^{(N)}(s))d s +
 \int_0^t \e ^{\Delta (t-s)}  B_1^{(N)}(\om^{(N)}(s))d s\\
 &\qquad +  \int_0^t \e ^{\Delta (t-s)} Q^{(N)}dW( s).
\end{align*}
Any weak limiting point satisfies therefore 
\eqref{020512}. To show  uniqueness we need the following.
\begin{lemma}
There exists a constant $C>0$ such that for all $w_0,w_1\in H$,  and $t\ge 0$, 
\begin{equation} \label{E27}
|\om(t;w_0)-\om(t;w_1)|\le |w_0-w_1|\exp\left\{ C\int_0^t \|\om(s;w_0)\|^2 d s\right\},\qquad \mathbb{P}-a.s. 
\end{equation}
\end{lemma}
\proof
Let $\rho(t):=\om(t;w_1)-\om(t;w_0)$ and $r(t):={\cal K}(\rho(t))$. From \eqref{E25} we conclude
\begin{equation} \label{030512}
\frac{d}{dt}|\rho(t)|^2=-2\|\rho(t)\|^2-2\langle( r(t)\cdot\nabla)\om(t;w_0),\rho(t)\rangle+2\langle( r(t,0)\cdot\nabla)\om(t;w_0),\rho(t)\rangle.
\end{equation}
To deal with the second term on the right hand side we use the following estimate. Suppose that 
$v={\cal K}(h)$. Then,  for any $r>0$ there exists a constant $C>0$ such that 
\begin{equation}
\label{020712}
|\langle( v\cdot\nabla)f,g\rangle|\le C\|f\||g|_{r}|h|,\quad\forall\, f\in V,\,g\in H_{r},\,h\in H
\end{equation}
and
\begin{equation}
\label{030712}
|\langle( v\cdot\nabla)f,g\rangle|\le C\|f\||g||h|_r,\quad\forall\,g\in H,\,f\in V,\,h\in H_r,
\end{equation}
see e.g. (6.10)  of \cite{foias}.
With these two inequalities in mind we conclude from \eqref{030512} that
\begin{eqnarray*}
\frac{d}{dt}|\rho(t)|^2&&\le 
-2\|\rho(t)\|^2+C\|\om(t;w_0)\||\rho(t)|_{1/2}|\rho(t)|\\
&&
\le -2\|\rho(t)\|^2+C_1\|\om(t;w_0)\|^2|\rho(t)|^2+2\|\rho(t)\|^2.
\end{eqnarray*}
By Gronwall's inequality we conclude then \eqref{E27}.
\qed

%

\section{Semimartingale estimates}

\label{secApB}

The following result comes from \cite{HM}, see Lemma 5.1. 
\begin{proposition}
\label{prop5.1}
	Let $\{U(t),\,t\ge0\}$ be a real-valued semi-martingale 
	$$
	dU(t)=F(t)dt +G(t)dw(t), \quad U(0)=u_0,
	$$
with  a standard Brownian motion $\{w(t),\,t\ge0\}$. Assume that there exist a process $\{Z(t),\,t\ge0\}$ and positive constants $b_1, b_2, b_3$, with $b_2 > b_3$, such that $F (t) \le  b_1 -b_2Z(t)$, $U(t) \le Z(t)$ and $G(t)^2 \le b_3Z(t)$, $\mathbb{P}$-a.s. Then
$$	
\bbE\exp\left\{U(t)+	\frac{b_2\e^{-b_2t/4}}{4}\int_0^tZ(s)ds\right\}
\le \frac{b_2 \exp\{2b_1/b_2\}}{b_2-b_3}	\exp\left\{u_0\e^{-(b_2 /2)t}\right\},\quad\forall\,t\ge0.
$$
\end{proposition}

Let  $U(t):=\nu|\om(t)|^2$. Using It\^o's formula and \eqref{E25a} we obtain
\begin{equation}
\label{051512}
dU(t)=\nu\left({\rm tr}\, Q^2-2\|\om(t)\|^2\right)d t+2\nu|Q\om(t)|\d w(t)
\end{equation}
for some adapted one dimensional standard Brownian motion $w$. 
Using Proposition \ref{prop5.1} with
$ Z(t) = \nu\|\om(t)\|^2$ and
$$
 b_1 =\nu{\rm tr}\,Q^2 ,\quad	b_2 =1,\quad	b_3 =4\nu\|Q\|^2.
$$
(see the proof of Proposition 5.2 of \cite{HM1} for details) we conclude that (cf (28) ibid.)
for $\nu_0:=1/(4\|Q\|)$ there exists a constant $C>0$ such that 
%
%
%
\begin{equation}
\label{010712}
\bbE\exp\left\{\nu|\om(t)|^2+\frac{\nu}{2e}\int_0^t\|\om(s)\|^2ds\right\}\le C_1\exp\left\{\nu|w|^2e^{-t/2}\right\},\quad\forall\,t\in[0,1],\,\nu\in[0,\nu_0].
\end{equation}

Using \eqref{051512} we can also repeat the proof of point 1) of  Lemma 4.10 from \cite{HM} and conclude that
there exist $\nu_0>0$ and a constant $C>0$ such that 
\begin{equation}
\label{010712-00}
\bbE\exp\left\{\nu\sup_{t\ge0}\left[|\om(t)|^2+\int_0^t\|\om(s)\|^2ds-t{\rm tr}\, Q^2\right]\right\}\le Ce_{\nu}(w),\quad\forall\,t\ge0,\,\nu\in[0,\nu_0].
\end{equation}

\end{document}